\documentclass[runningheads,envcountsect,envcountsame]{llncs}

\usepackage{xspace}
\usepackage{xcolor}
\usepackage{comment}
\usepackage[T1]{fontenc}
\usepackage{enumerate}
\usepackage{amssymb}
\usepackage{subfig}
\usepackage{algorithm}
\usepackage{algpseudocode2}
\usepackage{graphicx}

\usepackage{varioref}
\usepackage{hyperref}
\usepackage[capitalise,compress]{cleveref}

\usepackage{amsmath}
\spnewtheorem{RULE}{Reduction Rule}{\bfseries}{\upshape}
\crefname{RULE}{Reduction Rule}{Reduction Rules}
\spnewtheorem{BRANCHING}{Branching Rule}{\bfseries}{\upshape}
\crefname{BRANCHING}{Branching Rule}{Branching Rules}

\newcommand\abs[1]{\lvert #1\rvert}

\def\cobasis{10}
\def\cybasis{26}
\def\red{red\xspace}
\def\cc{\text{cc}}

\def\K_#1{{K_{#1}}}
\def\S_#1{\overline{K_{#1}}}

\newcommand\cC{\mathcal{C}}
\newcommand\cF{\mathcal{F}}
\newcommand\cH{\mathcal{H}}

\newcommand\cG{\mathcal{G}}

\newcommand\cP{\mathcal{P}}
\newcommand\cB{\mathcal{B}}

\newcommand\cO{\mathcal{O}}

\newcommand\cK{\mathcal{K}}

\newcommand{\YES}{\textsc{Yes}}
\newcommand{\NO}{\textsc{No}}
\newcommand{\OPT}{\operatorname{OPT}}
\newcommand{\classp}{\Phi_{\cP}}
\newcommand{\obd}{\mathcal{B}_{d+1, 2d-2}}
\newcommand{\cld}{(\cP \cap \mathcal{B}_{2,d})}

\newcommand{\FVS}{\textsc{Feedback Vertex Set}\xspace}
\newcommand{\WFVS}{\textsc{Weighted Feedback Vertex Set}\xspace}
\newcommand{\SFVS}{\textsc{Subset Feedback Vertex Set}\xspace}
\newcommand{\BGVD}{\textsc{Complete Block Vertex Deletion}\xspace}
\newcommand{\dBGD}{\textsc{Bound\-ed Block VD}\xspace}

\newcommand{\dCBGD}{\textsc{Bounded Cactus Graph VD}\xspace}
\newcommand{\dKBGD}{\textsc{Bounded Complete Block VD}\xspace}
\newcommand{\COC}{\textsc{Component Order Connectivity}\xspace}

\newcommand{\DHS}{\textsc{Diamond Hitting Set}\xspace}

\newcommand{\BPBVD}{\textsc{Bound\-ed $\mathcal{P}$-Block Vertex Deletion}\xspace}
\newcommand{\PBVD}{\textsc{$\mathcal{P}$-Block Vertex Deletion}\xspace}

\renewcommand{\leq}{\leqslant}
\renewcommand{\geq}{\geqslant}
\renewcommand{\le}{\leqslant}
\renewcommand{\ge}{\geqslant}

\let\Oldendproof\endproof%
\def\endproof{\qed\Oldendproof}%

\newcommand{\Problem}[4]{
\ \\
\noindent
\begin{tabular}{|p{0.99\linewidth}|}
  \hline
  #1 \hfill \textbf{Parameter: }#2\\
  \textbf{Input: }#3\\
  \textbf{Question: }#4\\
  \hline
\end{tabular}\\
}

\usepackage{chngcntr}
\usepackage{apptools}
\AtAppendix{\counterwithin{theorem}{section}}

\usepackage{etoolbox}
\usepackage{aliascnt}

\newtoggle{paper}
\toggletrue{paper}
\iftoggle{paper}{%
  \usepackage{a4}
  \spnewtheorem{RULE}{Reduction Rule}[section]{\bfseries}{\upshape}
  \spnewtheorem{BRANCHING}{Branching Rule}[section]{\bfseries}{\upshape}
}{}

\usepackage{thm-restate} %needs to be after \spnewtheorem's

\newcounter{Bew1}
\newcounter{Bew2}
\newcounter{Def1}

\newcommand{\appendixproof}[2]{
   \stepcounter{Bew1}
   \label{L\arabic{Bew1}}
   \gappto{\appendixProofText}{ \phantomsection \stepcounter{Bew2}\label{\arabic{Bew2}} \subsection{Proof of #1} #2}
	\vspace{-2pt}
}

\iftoggle{paper}{
\renewcommand{\appendixproof}[2]{#2}

}

\title{Parameterized vertex deletion problems for hereditary graph classes with a block property\thanks{All authors are supported by ERC Starting Grant PARAMTIGHT (No. 280152).}}
\titlerunning{Vertex deletion problems for graph classes with a block property} 

\author{\'Edouard Bonnet \and Nick Brettell \and O-joung Kwon \and D\'aniel Marx}
  \authorrunning{\'E. Bonnet \and N. Brettell \and O. Kwon \and D. Marx}
  
\institute{Institute for Computer Science and Control, Hungarian Academy of Sciences\\(MTA SZTAKI) \\%, Hungary.%\\
  \ \\
\href{mailto:edouard.bonnet@dauphine.fr}{\texttt{edouard.bonnet@dauphine.fr}},\,%
\href{mailto:nbrettell@gmail.com}{\texttt{nbrettell@gmail.com}},\,\\%
\href{mailto:ojoungkwon@gmail.com}{\texttt{ojoungkwon@gmail.com}},\,%
\href{mailto:dmarx@cs.bme.hu}{\texttt{dmarx@cs.bme.hu}}}

\begin{document}

\maketitle

\begin{abstract}
For a class of graphs $\mathcal{P}$, 
the \BPBVD problem asks, given a graph~$G$ on $n$ vertices and positive integers~$k$ and~$d$, 
whether there is a set~$S$ of at most $k$ vertices such that each block of $G-S$ 
%with at least one edge 
has at most $d$ vertices and is in $\mathcal{P}$. 
We show that when $\cP$ satisfies a natural hereditary property and is recognizable in polynomial time, 
\BPBVD can be solved in time $2^{\cO(k \log d)}n^{\cO(1)}$. 
When $\mathcal{P}$ contains all split graphs, we show that this running time is essentially optimal unless the Exponential Time Hypothesis fails.
On the other hand, if $\mathcal{P}$ consists of only complete graphs, or only cycle graphs and $K_2$,
then \BPBVD admits a $c^{k}n^{\cO(1)}$-time algorithm for some constant $c$ independent of $d$. 
We also show that \BPBVD admits a kernel with $\mathcal{O}(k^2 d^7)$ vertices.
\end{abstract}

\section{Introduction}\label{sec:introduction}

Vertex deletion problems are formulated as follows: given a graph~$G$ and a class of graphs~$\cG$, 
is there a set of at most $k$ vertices whose deletion transforms $G$ into a graph in $\cG$? 
A graph class~$\cG$ is \emph{hereditary} if whenever $G$ is in $\cG$, every induced subgraph $H$ of $G$ is also in $\cG$.
%Hereditary graph classes have been studied in depth for parameterized vertex deletion problems.
Lewis and Yannakakis~\cite{Lewis1980} proved that for every non-trivial hereditary graph class decidable in polynomial time, the vertex deletion problem for this class is NP-complete.
On the other hand, a class is hereditary if and only if it can be characterized by a set of forbidden induced subgraphs $\cF$, and
Cai~\cite{Cai1996} showed that if $\cF$ is finite, with each graph in $\cF$ having at most $c$ vertices, 
then there is an $\cO(c^k n^{c+1})$-time algorithm for the corresponding vertex deletion problem.
%Based on these results, considerable research effort has focused on obtaining fixed-parameter tractable (FPT) algorithms for various vertex deletion problems.

%A \emph{vertex deletion problem} poses an analogous question where the entire graph $G-X$ is required to have the property $\cP$.

%It follows that for a non-trivial property $\cP$ decidable in polynomial time, the block property problem on $\cP$ is %also
%NP-complete.
%Therefore, a natural next step is to consider the parameterized complexity of such block property problems.

%A property $\cP$ is hereditary if and only if it can be characterised by a set $\cF$ of forbidden induced subgraphs.
A \emph{block} of a graph is a maximal connected subgraph not containing a cut vertex.
Every maximal $2$-connected subgraph is a block, but a block may just consist of one or two vertices.
%Blocks are mostly $2$-connected subgraphs, but they may be connected subgraphs with at most $2$ vertices while $2$-connected subgraphs have at least $3$ vertices.
We consider vertex deletion problems for hereditary graph classes where all blocks of a graph in the class satisfy a certain common property.
It is natural to describe such a class by the set of permissible blocks~$\cP$.  For ease of notation, we do not require that $\cP$ is itself hereditary, but the resulting class, where graphs consist of blocks in $\cP$, should be.  %Thus, we To do this, we introduce the notion of a ``block-hereditary'' class, where if all blocks of a graph are in this class, the resulting graph is hereditary.
%We only require that graphs in $\cP$ satisfy a
%To describe such graph classes, we define the notion of a \emph{block-hereditary class}, which will be a class of allowed blocks.
To achieve this, we say that a class of graphs~$\cP$ is \emph{block-hereditary} if,
whenever $G$ is in $\cP$ and $H$ is an induced subgraph of $G$, every block of $H$ with at least one edge is isomorphic to a graph in $\cP$.
For a block-hereditary class of graphs $\cP$, we define $\classp$ as the class of all graphs whose blocks with at least one edge are in $\cP$.
Several well-known graph classes can be defined in this way.
For instance, 
a \emph{forest} is a graph in the class~$\Phi_{\{K_2\}}$,
a \emph{cactus graph} is a graph in the class~$\Phi_\cC$ where $\cC$ consists of $K_2$ and all cycles,
and a \emph{complete-block graph}\footnote{A \emph{block graph} is the usual name in the literature for a graph where each block is a complete subgraph.  However, since we are dealing here with both blocks and block graphs, to avoid confusion we instead use the term \emph{complete-block graph} and call the corresponding vertex deletion problem \BGVD.} is a graph in $\Phi_{\cK}$ where $\cK$ consists of all complete graphs.
We note that $\cC$ is not a hereditary class, but it is block-hereditary; this is what motivates our use of the term.
%
%
%
%For a block-hereditary class of graphs $\cP$, we define $\classp$ as the class of all graphs whose blocks with at least one edge are in $\cP$.
%Several well-known graph classes can be defined in this way. For instance, 
%a \emph{forest} is a graph in the class~$\classp$ where $\cP=\{K_2\}$, and
%a \emph{cactus graph} is a graph in the class~$\classp$ where $\cP$ consists of $K_2$ and all cycles, and a \emph{complete-block graph}~\footnote{It is usually called a block graph, and the corresponding vertex-deletion problem was called \textsc{Block Graph Vertex Deletion}~\cite{KimK2015,Agrawal}. Since we are dealing with various properties on blocks, these names make some confusion with other properties. To avoid confusion, we instead use a complete-block graph and \BGVD.} is a graph in the class~$\classp$ where $\cP$ consists of all complete graphs.
%We point out that the union of $K_2$ and all cycles is not hereditary, but block-hereditary, and this is the motivation for the definition of block-hereditary.
%forests, cactus graphs, and block graphs.
%In particular, a \emph{cactus graph} is a graph in the class $\classp$ where $\cP$ consists of $K_2$ and all cycles, and a \emph{block graph} is a graph in the class $\classp$ where $\cP$ consists of all complete graphs.

Let $\cP$ be a block-hereditary class such that $\classp$ is a non-trivial hereditary class. The result of Lewis and Yannakakis~\cite{Lewis1980} 
implies that the vertex deletion problem for $\classp$ is NP-complete. % in general.
We define
%Consider
the following parameterized problem for a fixed block-hereditary class of graphs $\cP$.
\Problem{\PBVD}{$k$}{A graph $G$ and a non-negative integer $k$.}{Is there a set $S \subseteq V(G)$ with $\abs{S} \leq k$ such that each block of $G-S$ with at least one edge is in $\cP$?}

This problem generalizes the well-studied parameterized problems \textsc{Vertex Cov\-er}, when $\cP=\emptyset$, and \textsc{Feedback Vertex Set}, when $\cP=\{K_2\}$.
Moreover, if %there is a finite set $\cF$ of $2$-connected graphs satisfying that a $2$-connected graph is in $\cP$ if and only if it has no induced subgraphs in $\cF$,
$\classp$ can be characterized by a finite set of forbidden induced subgraphs,
then Cai's approach~\cite{Cai1996} can be used
to obtain a fixed-parameter tractable (FPT) algorithm that runs in time $2^{\cO(k)}n^{\cO(1)}$.

%\noindent
%We consider this problem when parameterized by $k$, or, when $\cP$ is finite and every $P \in \cP$ has at most $d$ vertices, by $k$ and $d$.

In this paper, we are primarily interested in the variant of this problem where, additionally, the number of vertices in each block is at most $d$.  The value $d$ is a parameter given in the input.
\Problem{\BPBVD}{$d$, $k$}{A graph $G$, a positive integer $d$, and a non-negative integer $k$.}{Is there a set $S \subseteq V(G)$ with $\abs{S} \leq k$ such that each block of $G-S$ with at least one edge has at most $d$ vertices and is in $\cP$?}
%, and $\abs{V(B)} \leq d$

We also consider this problem when parameterized only by $k$.
\iftoggle{paper}{When $d=|V(G)|$, this problem is equivalent to \PBVD, so \BPBVD is NP-complete for any $\cP$ such that $\classp$ is a non-trivial hereditary class.}{}
When $d=1$, this problem is equivalent to \textsc{Vertex Cover}.
This implies that the \BPBVD problem is para-NP-hard when parameterized only by $d$.

The \BPBVD problem is also equivalent to \textsc{Vertex Cover} when $\cP$ is a class of edgeless graphs.
Since \textsc{Vertex Cover} is well studied,  we assume that $d\ge 2$, and focus on classes that contain a graph with at least one edge.
%We say that $\cP$ is \emph{non-degenerate} if it contains $K_2$.
We call such a class  \emph{non-degenerate}.
When $\cP$ is the class of all connected graphs with no cut vertices, we refer to \BPBVD as \dBGD.

\paragraph{\bf Related Work.}
%\subsection*{Related Work}

The analogue of \dBGD for connected components, rather than blocks, is known as \COC.
For this problem, the question is whether a given graph $G$ has a set of vertices~$S$ of size at most $k$ such that each connected component of $G-S$ has at most $d$ vertices.
Drange et al.~\cite{Drange2014} %studied the parameterized complexity of this problem, and
showed that \COC is $W[1]$-hard when parameterized by $k$ or by $d$, but FPT when parameterized by $k + d$, with an algorithm running in $2^{\cO(k \log d)}n$ time.

%A \emph{cactus graph} is a graph in the class~$\classp$ where $\cP$ consists of $K_2$ and all cycle graphs, and a \emph{block graph} is a graph in the class~$\classp$ where $\cP$ consists of all complete graphs.
Clearly, the vertex deletion problem for either cactus graphs, or complete-block graphs, is a specialization of \PBVD.
A graph is a cactus graph if and only if it does not contain a subdivision of the \emph{diamond}~\cite{MallahC1988}, the graph obtained by removing an edge from the complete graph on four vertices.
%It is known that a cactus graph if and only if it does not contain the subdivision of a diamond,
%where  a diamond is any subdivision of the graph consisting of three parallel edges~\cite{MallahC1988}.
For this reason, the problem for cactus graphs is known as \DHS.
%Clearly, the vertex deletion problem for either of cactus graphs and complete-block graphs is a specialization of \PBVD.
%It is known that a cactus graph if and only if it does not contain a subdivision of the diamond,
%where the diamond is the graph obtained from $K_4$ removing an edge~\cite{MallahC1988}.
%In this reason, the problem for cactus graphs is called as \DHS.
For block graphs, we call it \BGVD.
%The parameterized complexity of these problems has been studied recently~\cite{Misra2012, Kolay2015, KimK2015, Agrawal}.
%Both of these problems are specializations of \PBVD, but the former is typically known as \DHS, and the latter as \BGVD.
%These problems correspond to \PBVD where $\cP$ is the class of all cycles with $K_2$, or all complete graphs, respectively.
%The vertex deletion problem for cactus graphs is known as \DHS, %, and is equivalent to \PBVD when $\cP$ contains $K_2$ and all cycle graphs.
%and the vertex deletion problem for block graphs is known as \BGVD.
%These problems are easily seen to be specialisations of \PBVD.
General results %for \textsc{Planar-$\cF$ Deletion}
imply that there is a $c^k n^{\cO(1)}$-time algorithm for \DHS~\cite{Fomin2012,Joret2014,KLPRRSS13}, but an exact value for $c$ is not forthcoming from these approaches.
However, Kolay et al.~\cite{Kolay2015} obtained a $12^k n^{\cO(1)}$-time randomized algorithm.
For the variant where each cycle must additionally be odd
(that is, $\cP$ consists of $K_2$ and all odd cycles),
there is a $50^k n^{\cO(1)}$-time deterministic algorithm due to Misra et al.~\cite{Misra2012}.
%The vertex deletion problem for block graphs is known as \BGVD.
For \BGVD,
Kim and Kwon~\cite{KimK2015} showed that there is an algorithm that runs in $10^kn^{\mathcal{O}(1)}$ time, and there is a kernel with $\cO(k^6)$ vertices. 
Agrawal et al.~\cite{Agrawal} improved this running time to $4^kn^{\mathcal{O}(1)}$, and also obtained a kernel with $\mathcal{O}(k^4)$ vertices.

When considering a minor-closed class, rather than a hereditary class, the vertex deletion problem is known as \textsc{$\cF$-minor-free Deletion}.
Every \textsc{$\cF$-minor-free Deletion} problem has an $\cO(f(k) \cdot n^3)$-time FPT algorithm%, by the celebrated result of Robertson and Seymour
~\cite{Robertson1995}.
When $\cF$ is a set of connected graphs containing at least one planar graph, Fomin et al.~\cite{Fomin2012} showed there is a deterministic FPT algorithm for this problem running in time $2^{\cO(k)} \cdot \cO(n \log^2 n)$. %, and a randomised algorithm running in time $2^{\cO(k)}\cdot\cO(n)$.
One can observe that the class of all graphs whose blocks have size at most $d$ is closed under taking minors. %and by Graph Minor Theorem~\cite{RS2004}, 
%the class can be characterized by a finite set of forbidden minors. As one of the forbidden minors is a cycle of length $d+1$ which is planar, 
%in fact, \dBGD is a special case of \textsc{Planar-$\cF$ Deletion}~\cite{Fomin2012, KLPRRSS13}.
Thus, \PBVD has a single-exponential FPT algorithm and a polynomial kernel, when $\cP$ contains all connected graphs with no cut vertices and at most $d$ vertices.
However, it does not tell us anything about the parameterized complexity of \BPBVD, which we consider in this paper.

\paragraph{\bf Our Contribution.}
%\subsection*{Our Contribution}
The main contribution of this paper is the following:

% has $d$-bounded property and

\begin{theorem}\label{thm:main1}
Let $\cP$ be a non-degenerate block-hereditary class of graphs that is recognizable in polynomial time.
Then, \BPBVD
\begin{enumerate}[\rm (i)]
\item can be solved in $2^{\cO(k \log d)}n^{\cO(1)}$ time, and
\item admits a kernel with $\mathcal{O}(k^2 d^7)$ vertices.
\end{enumerate}
\end{theorem}

We will show that this running time is essentially optimal when $\classp$ is the class of all graphs, unless the Exponential Time Hypothesis (ETH)~\cite{MR1894519} fails.
One may expect that if the permissible blocks in $\cP$ have a simpler structure, then the problem becomes easier. 
%Thus, we also consider chordal graphs, split graphs, cycles, complete graphs as permissible blocks in the target graphs.
However, we obtain the same lower bound when $\classp$ contains all split graphs.
Since split graphs are a subclass of chordal graphs, the same can be said when $\classp$ contains all chordal graphs.

\begin{restatable}{theorem}{lbtheorem}\label{prop:dblockdeletion-lower-bound}
%\dBGD\ is not solvable in time $2^{o(k \log d)}$ unless the ETH fails.
Let $\cP$ be a block-hereditary class. If $\classp$ contains all split graphs,
then \BPBVD is not solvable in time $2^{o(k \log d)}$, unless the ETH fails.
\end{restatable}
\noindent
Formally, %The meaning of Theorem~\ref{prop:dblockdeletion-lower-bound} is that
there is no function $f(x)=o(x)$ such that there is a $2^{f(k\log d)}n^{\cO(1)}$-time algorithm for \BPBVD, unless the ETH fails.

%The meaning of Theorem~\ref{prop:dblockdeletion-lower-bound} is that there is no function $f(x)=o(x)$ such that there is a $2^{f(k\log d)}n^{O(1)}$-time algorithm.

%\begin{theorem}\label{thm:main3}
%\begin{enumerate}
%\item \dCBGD can be solved in time $\cO^*(\cybasis^k)$. 
%\item \dKBGD can be solved in time $\cO^*(\cobasis^k)$. 
%\end{enumerate}
%\end{theorem}

\iftoggle{paper}{
\begin{restatable}{proposition}{wonehardness}\label{prop:w1hardness}
Let $\cP$ be a block-hereditary class. If $\classp$ contains all split graphs,
then \BPBVD is $W[1]$-hard when parameterized only by $k$.
\end{restatable}

On the other hand, \BPBVD is FPT when parameterized only by $k$ if $\cP$ consists of all complete graphs, or if $\cP$ consists of $K_2$ and all cycles.
}{%
On the other hand, we show that the running time can be improved to $c^{k}n^{\cO(1)}$ for some $c$, independent of $d$, when $\cP$ consists of all complete graphs, or when $\cP$ consists of $K_2$ and all cycles.
}
We refer to these problems as \dKBGD and \dCBGD respectively.
%We refer to the \BPBVD problem as \dCBGD when $\cP$ consist of $K_2$ and all cycle graphs, and as \dKBGD when $\cP = \{K_m : m \geq 2\}$.
%These problems are variants of \DHS and \BGVD, respectively, with an additional constraint on the number of vertices in each block.

\begin{restatable}{theorem}{thmdKBGDiterative}
  \label{prop:dKBGDiterative}
  \dKBGD can be solved in time $\cO^*(\cobasis^k)$.
\end{restatable}

\begin{restatable}{theorem}{thmdCBGDiterative}
  \label{prop:dCBGDiterative}
  \dCBGD can be solved in time $\cO^*(\cybasis^k)$. 
\end{restatable}

\iftoggle{paper}{
When $d = |V(G)|$, these become $\cO^*(c^k)$-time algorithms for \BGVD and \DHS respectively.  In particular, the latter implies %the following:
that there
%\begin{corollary}
  %There
  is a deterministic FPT algorithm that solves \DHS, running in time $\cO^*(\cybasis^k)$. 
%\end{corollary}

  \medskip

  The paper is structured as follows.
  In the next section, we give some preliminary definitions.
  In \cref{sec:clustering}, we define $\cP$-clusters and $\cP$-clusterable graphs, and show that \PBVD can be solved in $\cO^*(4^k)$ time for $\cP$-clusterable graphs; in particular, we use this to prove \cref{thm:main1}(i).
  In \cref{sec:lower-bound}, we show that, assuming the ETH holds, this running time is essentially tight (\cref{prop:dblockdeletion-lower-bound}), and in \cref{sec:w1hardness} we prove \cref{prop:w1hardness}.
  In \cref{sec:single-exponential}, we use iterative compression to prove \cref{prop:dKBGDiterative,prop:dCBGDiterative}.
  Finally, in \cref{sec:polykernel}, we show that \BPBVD admits a polynomial kernel, proving \cref{thm:main1}(ii).  We also show that smaller kernels can be obtained for \dBGD, \dKBGD, and \dCBGD.
}{}

\section{Preliminaries}\label{sec:prelim}

%\emph{Blocks and paths.}
All graphs considered in this paper are undirected, and have no loops and no parallel edges. 
Let $G$ be a graph.
We denote by $N_G(v)$ the set of neighbors of a vertex $v$ in $G$, and let $N_G(S):=\bigcup_{v \in S} N_G(v) \setminus S$ for any set of vertices $S$.
%The vertex set of $G$ is denoted $V(G)$, and the edge set is denoted $E(G)$.
For $X\subseteq V(G)$, the \emph{deletion} of $X$ from $G$ is the graph obtained by removing $X$ and all edges incident to a vertex in $X$, and is denoted $G-X$.
For $x\in V(G)$, we simply use $G-x$ to refer to $G-\{x\}$. 
Let $\cF$ be a set of graphs; then $G$ is \emph{$\cF$-free} if it has no induced subgraph isomorphic to a graph in $\cF$.
For $n\ge 1$, the complete graph on $n$ vertices is denoted $K_n$. %, and, for $n \ge 3$, $C_n$ denotes the cycle graph of $n$ vertices.

A vertex $v$ of $G$ is a {\em cut vertex} if the deletion of $v$ from $G$ increases the number of connected components. 
We say $G$ is \emph{biconnected} if it is connected and has no cut vertices.
A \emph{block} of $G$ is a maximal biconnected subgraph of $G$.
The graph $G$ is \emph{$2$-connected} if it is biconnected and $\abs{V(G)}\ge 3$.
In this paper we are frequently dealing with blocks, so %it is often more natural to consider a biconnected subgraph, rather than a $2$-connected subgraph.
the notion of being biconnected is often more natural than that of being $2$-connected.
The \emph{block tree} of $G$ is a bipartite graph $B(G)$ with bipartition $(\cB, X)$, where $\cB$ is the set of blocks of $G$, $X$ is the set of cut vertices of $G$, and a block~$B \in \cB$ and a cut vertex~$x\in X$ are adjacent in $B(G)$ if and only if $B$ contains $x$.
\iftoggle{paper}{%
A block $B$ of $G$ is a \emph{leaf block} if $B$ is a leaf of the block tree $B(G)$.  Note that a leaf block has at most one cut vertex.

For $u,v \in V(G)$, a \emph{$uv$-path} is a path beginning at $u$ and ending at $v$.
For $X \subseteq V(G)$, an \emph{$X$-path} is a path beginning and ending at distinct vertices in $X$, with no internal vertices in $X$.
For $v \in V(G)$ and $X \subseteq V(G)$, a \emph{$(v,X)$-path} is a path beginning at $v$, ending at a vertex $x \in X$, and with no internal vertices in $X$.
The \emph{length} of a path $P$, denoted $l(P)$, is the number of edges in $P$.
A path is \emph{non-trivial} if it has length at least two.
}{}
 
\emph{Parameterized Complexity.}
A parameterized problem $Q\subseteq \Sigma^* \times N$ is \emph{fixed-parameter tractable} (\emph{FPT}) if there is an algorithm that decides whether $(x,k)$ belongs to $Q$ in time $f(k)\cdot \abs{x}^{\mathcal{O}(1)}$ for some computable function $f$. Such an algorithm is called an {\em FPT algorithm}. 
A parameterized problem is said to admit a \emph{polynomial kernel} if there is a polynomial time algorithm in $\abs{x}+k$, called a \emph{kernelization algorithm}, that reduces an input instance into an instance with size bounded by a polynomial function in $k$, while preserving the \YES\ or \NO\ answer.

\section{Clustering}\label{sec:clustering}

Agrawal et al.~\cite{Agrawal} described an efficient FPT algorithm for \BGVD\ using a two stage approach.
Firstly, small forbidden induced subgraphs are eliminated using a branching algorithm.
More specifically, for each diamond or cycle of length four, %or diamond, a graph obtained by removing an edge from $K_4$,
at least one vertex must be removed in a solution, so there is a branching algorithm that runs in $\cO^*(4^k)$ time.
%More specifically, if $X \subseteq V(G)$ such that $G[X]$ is isomorphic to a diamond or $C_4$, 
%a solution must contain a vertex in $X$, so one can branch on the inclusion of each vertex in the solution, leading to a
%Firstly, small obstructions (diamonds and $C_4$'s) are removed by branching. %; doing so takes
%where at most $k$ vertices can be removed in a solution.
%At the completion of this procedure, if a solution $S$ is found, the graph $G-S$ can be covered by maximal cliques, where 
The resulting graph has the following structural property: any two distinct maximal cliques have at most one vertex in common.
Thus, in the second stage, it remains only to eliminate all cycles not fully contained in a maximal clique, so the problem can be reduced to an instance of \WFVS.
%We adapt this approach in order to obtain an FPT algorithm for other block property problems, in particular \dBGD.  
%in particular \dCBGD (and P-clusterable properties).
We generalize this process and refer to it as ``clustering'', where the ``clusters'', in the case of \BGVD, are the maximal cliques.
We use this to obtain an algorithm for \BPBVD in \cref{clustering-BPBVD}.

\subsection{$\cP$-clusters}

%%We say a block is \emph{trivial} if it has no edges.
%%A class of graphs $\cP$ is \emph{block-hereditary} if whenever a graph $G$ is in $\cP$, then every non-trivial block of an induced subgraph of $G$ is in $\cP$.
%%%Definition
Let $\cP$ be a block-hereditary class of graphs.
We may assume that $\cP$ contains only biconnected graphs; otherwise there is some block-hereditary $\cP'$ such that $\cP' \subset \cP$ and $\Phi_{\cP'} = \classp$.
%and with the property that if a graph $G$ is in $\cP$, then every non-trivial block of an induced subgraph of $G$ is in $\cP$.
%We also assume, in what follows, that $K_2$ has property~$\cP$; otherwise, $\cP$ coincides with the trivial class of empty (edgeless) graphs.
Let $G$ be a graph.
A \emph{$\cP$-cluster} of $G$ is a maximal induced subgraph~$H$ of $G$ with the property that $H$ is isomorphic either to $K_1$ or a graph in $\cP$.
%with $H \in \cP$ that
%is either $2$-connected or isomorphic to $K_2$, and is maximal.
%When $\cP$ is obvious from context, we simply call $H$ a \emph{cluster}.
We say that $G$ is \emph{$\cP$-clusterable} if for any distinct $\cP$-clusters $H_1$ and $H_2$ of $G$, we have $|V(H_1) \cap V(H_2)| \leq 1$.
%We also say that a class of graphs $\cG$ is \emph{$\cP$-clusterable} when $G$ is $\cP$-clusterable for every $G \in \cG$.
For a $\cP$-clusterable graph,
if $v \in V(G)$ is contained in at least two distinct $\cP$-clusters, then $v$ is called an \emph{external} vertex.

%Clearly, each edge of $G$ is contained in precisely one cluster of $G$.
%; otherwise, $v$ is an \emph{internal} vertex.

The following property of $\cP$-clusters is essential.
%Let $G$ be a graph and let $C$ be a cycle subgraph of $G$.
We say that $X \subseteq V(G)$ \emph{hits} a cycle $C$ if $X \cap V(C) \neq \emptyset$, and a cycle $C$ is \emph{contained} in a $\cP$-cluster of $G$ if $V(C) \subseteq V(H)$ for some $\cP$-cluster $H$ of $G$.

%\todo{Need to clarify whether assuming $G$ is $\cP$-clusterable is necessary or not. Daniel thought that this assumption is redundant. }
%Indeed it is not necessary; but is for the subsequent proposition.  Have fixed the statement of this lemma accordingly --Nick.
\begin{lemma}
  \label{dblock-clustering}
  Let  $\cP$ be a non-degenerate block-hereditary class of graphs,
  let $G$ be a graph, and let $S\subseteq V(G)$.
  Then %, for $S\subseteq V(G)$,
  $G-S\in \classp$ if and only if $S$ hits every cycle not contained in a $\cP$-cluster of $G$.
\end{lemma}
\appendixproof{\cref{dblock-clustering}}
{
\begin{proof}
  Suppose $G-S\in \classp$ and there exists a cycle $C$ of $G-S$ that is not contained in a $\cP$-cluster of $G$.  
  As $G-S\in \classp$ and every cycle is biconnected, $G[V(C)]$ is in $\cP$.
  Thus, there exists a $\cP$-cluster of $G$ that contains $C$ as a subgraph; a contradiction.
  For the other direction, suppose $S$ hits every cycle not contained in a $\cP$-cluster, and let $B$ be a block of $G-S$. 
  It is sufficient to show that $B\in \cP$.
  If $B$ is not contained in a $\cP$-cluster, then there are distinct vertices $v_1$ and $v_2$ in $B$ such that $v_1\in V(P_1)\setminus V(P_2)$ and $v_2\in V(P_2)\setminus V(P_1)$ for distinct $\cP$-clusters $P_1$ and $P_2$.
  Since $K_2\in \cP$, we may assume that $B$ is not isomorphic to $K_2$.
  Thus, as $B$ is biconnected, there is a cycle containing $v_1$ and $v_2$ in $G-S$; a contradiction.
  We deduce that $B$ is contained in a $\cP$-cluster of $G$, so $B\in \cP$.
\end{proof}
}

We now show that \textsc{$\cP$-Block Vertex Deletion} can be reduced to \SFVS if the input graph is $\cP$-clusterable.
The \SFVS\ problem can be solved in time $\cO^*(4^k)$~\cite{wahlstrom2014}. %[Theorem~5.1]

\Problem{\SFVS}{$k$}{A graph $G$, a set $X \subseteq V(G)$, and a non-negative integer~$k$.}{Is there a set $S \subseteq V(G)$ with $|S| \leq k$ such that no %simple
cycle in $G - S$ contains a vertex of $X$?}

\begin{proposition}
  \label{killcycles1}
  Let  $\cP$ be a non-degenerate block-hereditary class of graphs recognizable in polynomial time.
  Given a $\cP$-clusterable graph~$G$ together with the set of $\cP$-clusters of $G$, and a non-negative integer $k$, 
  there is an $\cO^*(4^k)$-time algorithm that determines whether 
  there is a set $S \subseteq V(G)$ with $\abs{S} \le k$ such that $G-S\in \classp$.
\end{proposition}
\appendixproof{\cref{killcycles1}}
{
\begin{proof}
By \cref{dblock-clustering}, it is sufficient to determine whether $G$ contains a set $S \subseteq V(G)$ of size at most $k$ that hits 
 all cycles not contained in a $\cP$-cluster. To do this, we perform a reduction to \SFVS.
  %Note that we can obtain the set of $\cP$-clusters of $G$ %, and hence the external vertices of $G$,
  %in time polynomial in $V(G)$, by Lemma~\ref{clusterable-clusters}.
  We construct a graph $G'$ from $G$ as follows.
  Let $X \subseteq V(G)$ be the set of external vertices of $G$.
  For each $x \in X$,
  let $\{H_1,H_2,\dotsc,H_m\}$ be the set of $\cP$-clusters of $G$ that $x$ is contained in,
  and introduce $m$ vertices $v(x,H_i)$ for each $i \in \{1,2,\ldots,m\}$.
  Then, do the following for each $x \in X$.
  %Let $N'(H_i) := (N_{H_i}(x) \setminus X) \cup \{v(x,H_i) : x \in V_{H_i} \cap X\}$.
  Recall that $N_{H_i}(x)$ is the set of neighbors of $x$ contained in $H_i$, and set $N'_{H_i}(x) := (N_{H_i}(x) \setminus X) \cup \{v(y,H_i) : y \in N_{H_i}(x) \cap X\}$.
  %Note that $(V_1,\dotsc,V_m)$ is a partition of the neighbors of $v$, since $G$ is $\cP$-clusterable.
  Now remove all edges incident with $x$, and, for each $i \in \{1,2,\dotsc,m\}$, make $v(x,H_i)$ adjacent to each vertex in $N'_{H_i}(x) \cup \{x\}$. 
  %Note that $G$ can be obtained from $G'$ by contracting each edge incident to a vertex $x \in X$, labelling the resulting vertex $x$.
  This completes the construction of $G'$.
  See \cref{trans} for an example of this construction.
  We claim that $(G', X, k)$ is a \YES-instance for \SFVS\ 
  if and only if
  $G$ has a set of at most $k$ vertices that hits every cycle not contained in a $\cP$-cluster of $G$.
  \begin{figure}
    \centering
    \subfloat[$G$]{
      \includegraphics[scale=0.52]{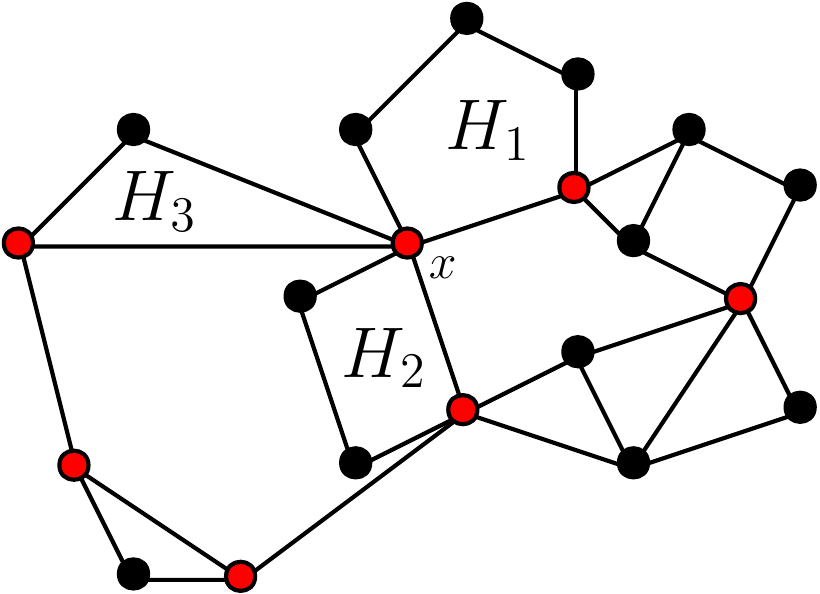}
    }
    \qquad
    \subfloat[$G'$]{
      \includegraphics[scale=0.41]{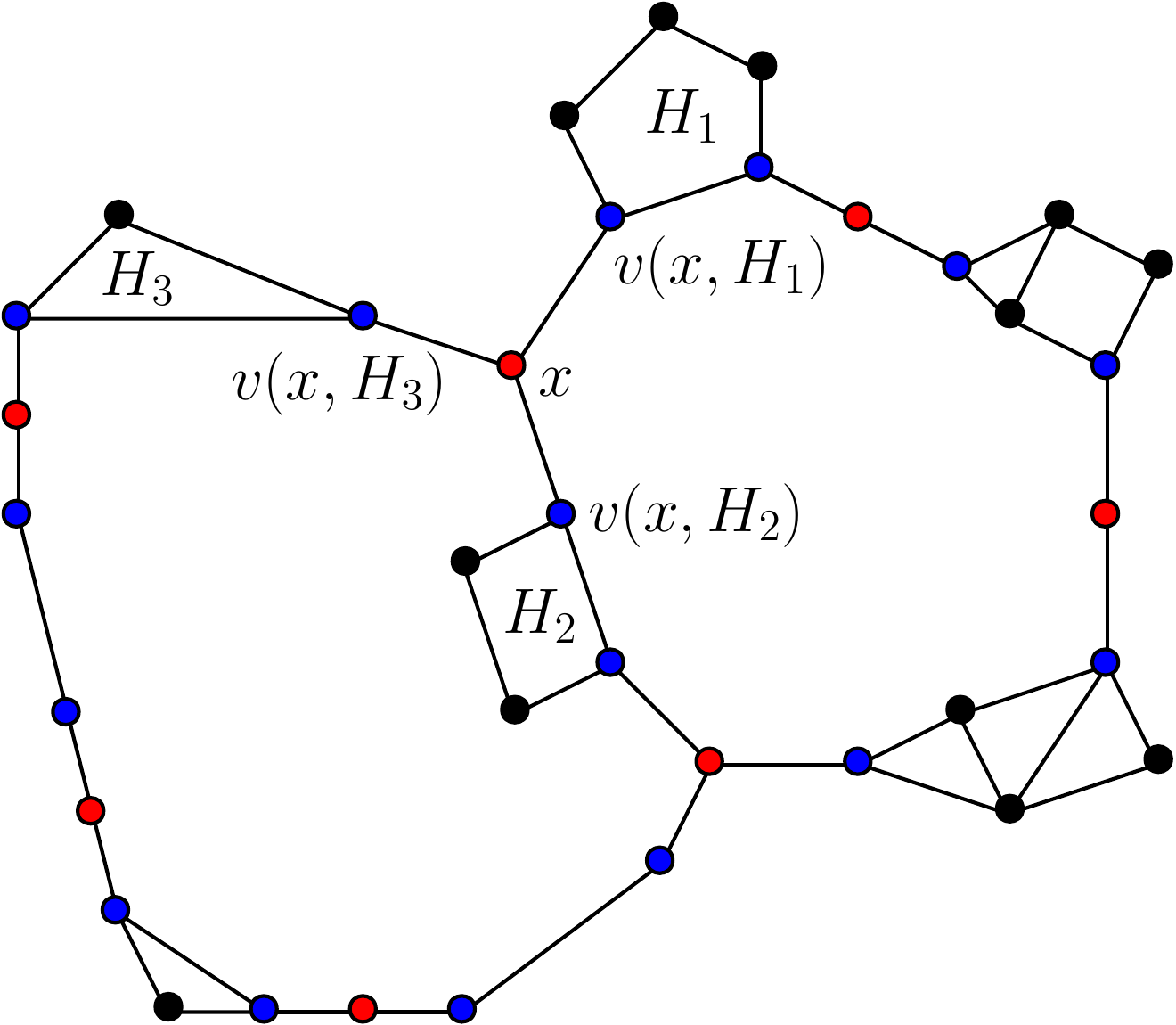}
    }
    \caption{Example construction of $G'$ from $G$, as described in \cref{killcycles1}.} \label{trans}
  \end{figure}

%\todo{here, need to use the property where two $\cP$-clusters share at most one vertex}
  %I believe this is now fixed. --Nick.
  Let $S \subseteq V(G)$ such that every cycle of $G-S$ is contained in a $\cP$-cluster of $G$.
  Towards a contradiction, suppose there is a cycle $C'$ of $G'-S$ containing at least one vertex $v\in X$.
  %Since $G$ can be obtained from $G'$ by contracting each edge incident to a vertex $x \in X$,
  Note that $G$ can be obtained from $G'$ by contracting each edge incident to a vertex $x \in X$, where the resulting vertex is labeled $x$.
  Thus, we can likewise obtain a cycle $C$ of $G-S$ by contracting each edge of $C'$ incident with a vertex $x$ in $X$, labeling the resulting vertex $x$, and relabeling any remaining vertices not in $V(G)$ by their unique neighbor in $G'$ that is a member of $X$.
%  Each vertex in $C' \setminus V(G)$ is adjacent to an expansion vertex, 
  Suppose $v$ is adjacent to $u$ and $w$ in $C$. 
  Then, by the construction of $G'$, %and since $G$ is $\cP$-clusterable,
  $u \in V(H)$ and $w \in V(H')$ for distinct $\cP$-clusters $H$ and $H'$ of $G$.
  Clearly, $C$ is not contained in a $\cP$-cluster of $G$ unless $u$ and $w$ are (not necessarily distinct) external vertices; but this implies that $H$ and $H'$ share at least two vertices $u$ and $v$, contradicting the fact that $G$ is $\cP$-clusterable.
  %So $C$ is cycle of $G-S$ not contained in a $\cP$-cluster of $G$; a contradiction.
  We deduce that no cycle of $G'-S$ contains a vertex in $X$, as required.
%  We conclude that $S$ is a solution to \SFVS\ on the instance $(G', X, k)$.

  Now let $S'$ be a solution to \SFVS\ on $(G', X, k)$.
  By the construction of $G'$, each vertex in $V(G') \setminus V(G)$ is adjacent to precisely one vertex in $X$.
  Let $U$ be the set of vertices in $X$ adjacent to a vertex in $S' \setminus V(G)$ and set $S := (S' \cap V(G)) \cup U$.
  Then $\abs{S} \leq \abs{S'}$ and $S \subseteq V(G)$.
  We claim that every cycle of $G-S$ is contained in a $\cP$-cluster.
  Suppose not; let $C$ be
  %Towards a contradiction, suppose $C$ is
  a cycle of $G-S$ not contained in a $\cP$-cluster.
  Note that for each $v\in V(C)\cap X$, the vertex $v$ and its neighbors in $G'$ are not in $S'$.
  Thus, we obtain a cycle $C'$ of $G'-S'$ from $C$ by performing one of the two following operations for each vertex $v \in V(C) \cap X$, where $u$ and $w$ are the two neighbors of $v$ in $C$:
  \begin{enumerate} 
    \item If there is a $\cP$-cluster $H$ of $G$ for which $\{u,v,w\} \subseteq V(H)$, then relabel $v$ in $C$ with the vertex in $V(G')\setminus V(G)$ adjacent to $\{v\} \cup (N_G(v)\cap V(H))$.
	\item Otherwise, for some $\cP$-clusters $H_1$ and $H_2$ of $G$, we have that $u \in V(H_1) \setminus V(H_2)$, $v \in V(H_1) \cap V(H_2)$, and $w \in V(H_2) \setminus V(H_1)$.  In this case, we subdivide $uv$ and $vw$ in $C$, labeling the new vertices $v_1$ and $v_2$ respectively, where $v_i$ is the vertex in $V(G')\setminus V(G)$ adjacent to $\{v\}\cup (N_G(v)\cap V(H_i))$, for $i\in \{1,2\}$.
  \end{enumerate}
  As $C$ is a cycle of $G-S$ not contained in a $\cP$-cluster, at least one vertex $x$ in $V(C)\cap X$ has its neighbors in $C$ in distinct $\cP$-clusters.
  By the second operation above, $x \in X$ is a vertex of $C'$; a contradiction.
  %$C'$ is a cycle of $G'-S'$ containing at least one vertex of $X$; a contradiction.
  We conclude that every cycle of $G-S$ is contained in a $\cP$-cluster.
  %thus completing the proof.
  \end{proof}
}
%%%%  
%%%%  Let $\{v_1,v_2,\dotsc,v_x\}$ be the external vertices of $G$ in $V(C)$, where $G'$ was obtained by the expansion of $G$ at $v_i$ induced by $(V_{1,i}, \dotsc, V_{s_i,i})$ for each $i \in \{1,\dotsc,x\}$.  For each $v_i$, %with $i \in \{1,\dotsc,x\}$,
%%%%  %Let $v_i \in V(C)$ be an external vertex of $G$, and say the expansion of $G$ at $v_i$ is induced by $(V_{1,i}, \dotsc, V_{s_i,i})$.
%%%%  the neighbours $u_i,w_i$ of $v_i$ in $C$ are either in the same $\cP$-cluster of $G$, or in distinct $\cP$-clusters.  If they are in the same $\cP$-cluster, $u_i,w_i \in V_{1,i}$ say, then the path $u_iv_iw_i$ in $C$ is replaced by $u_iv_{1,i}w_i$ in $C'$.
%%%%  If $u_i$ and $w_i$ are in different $\cP$-clusters, say $u_i \in V_{1,i}$ and $w_i \in V_{2,i}$, then the path $u_iv_iw_i$ in $C$ is replaced by $u_iv_{1,i}v_iv_{2,i}w_i$ in $C'$.
%%%%  Since $C$ is not contained in a $\cP$-cluster, there are at least two external vertices in $C$ whose neighbours in $C$ are in distinct $\cP$-clusters.
%%%%  Hence, $C'$ contains at least two expansion vertices.
%%%%  But no cycle of $G'-T'$ contains an expansion vertex, so $G'-T'$ does not contain $C'$.
%%%%  That is, $C'$ contains at least one vertex, $v_{1,i}$ say, that is in $T'\setminus V(G)$. 
%%%%  But any such vertex $v_{1,i}$ has an expansion vertex neighbour $v_i$ in $U \subseteq T$, by construction, and $v_i$ is in $C$, so $C$ is not a cycle of $G-T$; a contradiction.
%%%%  So any cycles in $G-T$ are contained in a $\cP$-cluster, thus completing the proof.

  By \cref{killcycles1}, the \PBVD problem admits an efficient FPT algorithm provided we can reduce the input to $\cP$-clusterable graphs.
  In the next section, we show that this is possible for any finite block-hereditary $\cP$ where the permissible blocks in $\cP$ have at most $d$ vertices.
  In particular, we use this to show there is an $\cO^*(2^{\cO(k \log d)})$-time algorithm for \BPBVD.
  %In Section~\ref{subsec:clustering}, we give a way to obtain a $\cP$-clusterable graph from any instance by removing $2$-connected subgraphs with between $d+1$ and $2d-2$ vertices.

\subsection{An FPT Algorithm for \BPBVD}\label{subsec:clustering}

\label{clustering-BPBVD}

In this section we describe an FPT algorithm for \BPBVD using the clustering approach.
For positive integers $x$ and $y$, let $\cB_{x,y}$ be the class of all biconnected graphs with at least $x$ vertices and at most $y$ vertices.  When $x > y$, $\cB_{x,y} = \emptyset$.

\iftoggle{paper}{}{The next three proofs are in the appendix.}
\begin{lemma}
\label{boundedclusterlemma}
Let $\cP$ be a non-degenerate block-hereditary class, and let $d \geqslant 2$ be an integer.
If a graph $G$ is $\obd$-free and 
$(\cB_{2,d} \setminus \cP)$-free,
then $G$ is $(\cP \cap \cB_{2,d})$-clusterable.
\end{lemma}

\appendixproof{Lemma~\ref{boundedclusterlemma}}
{
\begin{proof}
  %Clearly the lemma holds when $d \leq 2$, so we may assume that $d \geq 3$.
  Suppose $G$ has distinct $(\cP \cap \cB_{2,d})$-clusters $H_1$ and $H_2$ such that $|V(H_1) \cap V(H_2)| \geq 2$.
  Set $G':=G[V(H_1) \cup V(H_2)]$.
  By the maximality of $(\cP \cap \cB_{2,d})$-clusters, $V(H_1) \setminus V(H_2)$ and $V(H_2) \setminus V(H_1)$ are non-empty, so $|V(G')| \geq 4$.
  The graph $G' - v$ is connected for every $v \in V(G')$, so $G'$ is $2$-connected.
  Since $|V(G')| \leq 2d-2$, but $G$ is $\obd$-free, $|V(G')| \leq d$.
  Hence, by the maximality of $(\cP \cap \cB_{2,d})$-clusters, $G' \notin \cP$, %; a contradiction.
  which contradicts the fact that $G$ is $(\cB_{2,d} \setminus \cP)$-free.
\end{proof}
}

\begin{proposition} %[Bounded $\cP$-clustering algorithm]
  \label{finddblockalgo}
   Let $d\geqslant 2$ be an integer, and let $\cP$ be a non-degenerate block-hereditary class recognizable in polynomial time.
  There is a polynomial-time algorithm that, given a graph~$G$, either
  \begin{enumerate}[\rm (i)]
    \item outputs an induced subgraph of $G$ in $\cB_{2,d} \setminus \cP$, or
    \item outputs an induced subgraph of $G$ in $\obd$, or
    \item correctly answers that $G$ is %$\cP'$-clusterable, where $\cP' = \cP \cap \cB_{2,d}$.%. % is the restriction of $\cP$ to graphs with at most $d$ vertices.
      $((\cB_{2,d} \setminus \cP) \cup \obd)$-free.
%, {\color{red}and outputs the set of all $\cP'$-clusters.}%
  %\todo{I think we should instead try to prove the more general result that there is a polynomial time algorithm that, given a $\cP$-clusterable graph $G$, outputs all $\cP$-clusters of $G$.  But we need to check this is doable.  --Nick.}%
      \label{case3}%
  \end{enumerate}
\end{proposition}

\appendixproof{\Cref{finddblockalgo}}
{
\iftoggle{paper}{}{%
For $u,v \in V(G)$, a \emph{$uv$-path} is a path beginning at $u$ and ending at $v$.
For $X \subseteq V(G)$, an \emph{$X$-path} is a path beginning and ending at distinct vertices in $X$, with no internal vertices in $X$.
For $v \in V(G)$ and $X \subseteq V(G)$, a \emph{$(v,X)$-path} is a path beginning at $v$, ending at a vertex $x \in X$, and with no internal vertices in $X$.
The \emph{length} of a path $P$, denoted $l(P)$, is the number of edges in $P$.
A path is \emph{non-trivial} if it has length at least two.
}

\begin{proof}
  If %$\cP \subseteq \cB_{2,2}$ or  %%Not needed I think
  $d \leq 2$, then (\ref{case3}) holds trivially, so we may assume otherwise.
%
  %Firstly,
  We show that there is a polynomial-time algorithm \textsc{FindObstruction} that finds an induced subgraph of $G$ that is either in $\cB_{2,d} \setminus \cP$ or in $\obd$, if such an induced subgraph exists.  For brevity, we refer to either type of induced subgraph %, that is not isomorphic to $K_2$,
  as an \emph{obstruction}.
  In the case that no obstruction exists, %we wish to find the set of all $\cP$-clusters.  Secondly, we present a polynomial-time procedure that does this. %given a graph $\cP$-clusterable graph with no obstructions, it returns the set of all $\cP$-clusters.
  then $G$ is $(\cP \cap \cB_{2,d})$-clusterable, by Lemma~\ref{boundedclusterlemma}.

  First, we give an informal description of \textsc{FindObstruction} (\cref{foalg}).  We incrementally construct a biconnected induced subgraph $G[X]$, starting with $G[X]$ as the shortest cycle of $G$, by adding the vertices of a non-trivial $X$-path to $X$.  If, at any increment, $G[X]$ is an obstruction, then we return $G[X]$.  Otherwise, we eventually have that $G[X]$ is in $\cB_{2,d} \cap \cP$, but the union of $G[X]$ and any non-trivial $X$-path is not in $\obd$.
  Now, if there is a non-trivial $X$-path that together with a path in $G[X]$ forms a cycle of length at most $2d-2$, then this cycle is an obstruction, and we return it.  Otherwise, no obstruction intersects $G[X]$ in any edges, so we remove the edges of $G[X]$ from consideration and repeat the process.

\begin{algorithm}%[htp]
  \caption{\textsc{FindObstruction($G,d$)}}\label{foalg}
\begin{algorithmic}[1]
\Statex \textbf{Input:} A graph $G$ and an integer $d > 2$.
\Statex \textbf{Output:} An induced subgraph of $G$ that is either in $\cB_{2,d} \setminus \cP$ or in $\obd$; or \NO, if no such induced subgraph exists.
\State Set $\cH := \emptyset$.
\While {the shortest cycle in $G-\bigcup_{H\in \cH} E(H)$ has length at most $2d-2$,} \label{inwhile}
\State Set $G':=G-\bigcup_{H\in \cH} E(H)$. \label{defineg}
\State Let $X$ be the vertex set of the shortest cycle in $G'$.
\If {$\abs{X} \geq d+1$ or $G[X]\notin \cP$}
\State \Return $G[X]$. \label{earlyreturn}
\EndIf \Comment {$G[X]\in \cP$ and $\abs{X} \leq d$}
%\State Let $P$ be the shortest $X$-path in $G'$.
%\While {$|X \cup V(P)| \leq d$,}
\While {there is a non-trivial $X$-path $P$ in $G'$ such that $\abs{X \cup V(P)} \leq d$,}\label{nontrivialpath1}
\State Set $X := X \cup V(P)$.
\If {$G[X]\notin \cP$}
\State \Return $G[X]$. \label{earreturn1}
\EndIf
%\State Set $P$ to be the shortest $X$-path in $G'$.
\EndWhile\label{nontrivialpath2}
%\If {$X \neq V(G')$}
\Comment {$G[X]\in \cP$ and $\abs{X \cup V(P)} \geq d+1$ for any non-trivial $X$-path $P$} %in $G'$}
\If {there is a non-trivial $X$-path $P$ in $G'$ such that $\abs{X \cup V(P)} \leq 2d-2$,}
%\If {$|X \cup V(P)| \leq 2d-2$,}
\State \Return $G[X \cup V(P)]$. \label{earreturn2}
\ElsIf{for some distinct $u,v \in X$, there is a $uv$-path $P$ in $G[X]$ and a $uv$-path~$P'$ in $G'-E(G[X])$ with $l(P) + l(P') \leq 2d-2$,}\label{intersectioncheck}
\State \Return $G[V(P) \cup V(P')]$. \label{cyclereturn}
\EndIf
%\EndIf
\Comment{No obstruction meets $E(G[X])$}
\State Add $G[X]$ to $\cH$. \label{addcluster}
\EndWhile
%\If{$K_2 \notin \cP$ and $G$ has an edge $xy$}\label{handleK2s}
%\State \Return $G[\{x,y\}]$ %\Comment{Graphs in $\cP$ contain only parallel edges}
%\EndIf \label{handleK2e}
\State \Return \NO.
\end{algorithmic}
\end{algorithm}

Now we prove the correctness of the algorithm.
Since $\cP$ is non-degenerate, $K_2$ is not an obstruction.

Consider an induced subgraph $H:=G[X]$ added to $\cH$ at line~\ref{addcluster}.  As the edges of $H$ are excluded in future iterations, we will show that these edges do not meet the edges of any obstruction.
That is, we claim that for every $Y \subseteq V(G)$ such that
$G[Y]$ is an obstruction, $E(H) \cap E(G[Y]) = \emptyset$. %, that is, $H$ does not destroy any subgraph of $G$ in $\obd$.
Towards a contradiction, suppose $G[Y]$ is an obstruction for some $Y \subseteq V(G)$, and $uv \in E(H) \cap E(G[Y])$.
Clearly $\abs{X} \leq d$, and $G[X] \in \cP$.
Note also that $Y \nsubseteq X$, since $\abs{Y} \leq d$ implies that $G[Y] \notin \cP$, but $G[X] \in \cP$.
Since $G[Y]$ is biconnected and not isomorphic to $K_2$, it has a cycle subgraph $C$ that contains the edge $uv$.
Let $P$ be a non-trivial $X$-path in $C$. % of length at least two.
%Such a path exists since $Y \nsubseteq X$.
Since there is a path in $H$ of length at most $d-1$ between any two distinct vertices in $X$, it follows that $P$ has length more than $d-1$, otherwise there is a path satisfying the conditions of
line~\ref{intersectioncheck}.  Since $\abs{Y} \leq 2d-2$, we deduce that there are no two distinct $X$-paths in $C$.
So $C$ is the union of $P$ and a path contained in $H$.  Thus $P$ is a $uv$-path satisfying the conditions of line~\ref{intersectioncheck}; a contradiction.
%Since line~\ref{addcluster} is the only place where $\cH$ is modified,
This verifies our claim.

Now we show that if $G$ contains an obstruction, then the algorithm outputs an obstruction.
Suppose that $G$ contains an induced subgraph $J \in (\cB_{2,d} \setminus \cP) \cup \obd$.  
Since $\abs{V(J)} \leq 2d-2$ and $J$ is biconnected, but not isomorphic to $K_2$, it contains a cycle of length at most $2d-2$.
So execution reaches line~\ref{inwhile} and, from the previous paragraph, $J$ is a subgraph of in $G'$ as given in line~\ref{defineg}.
It is clear that the flow of execution will reach line~\ref{earlyreturn}, line~\ref{earreturn1}, line~\ref{earreturn2}, line~\ref{cyclereturn}, or line~\ref{addcluster}.
Clearly, a graph returned at line~\ref{earlyreturn} is an obstruction.
A graph returned at either line~\ref{earreturn1} or line~\ref{earreturn2} is $2$-connected, since it has an obvious open ear decomposition (see, for example, \cite[Theorem~5.8]{bondyAndMurty}), and thus is easily seen to be an obstruction.
If a graph is returned at line~\ref{cyclereturn}, it is a cycle, and since $l(P') \geq d$, such a graph is also in $\obd$.
By the previous paragraph, if execution reaches line~\ref{addcluster}, then $E(G[X])$ does not meet $E(J)$, so execution will loop, with $\abs{E(G')}$ strictly smaller in the following iteration.
Thus, eventually the algorithm will find either $J$, or another obstruction.

It remains to prove that the algorithm runs in polynomial time.
We observe that there will be at most $\cO(m)$ loops of the outer `while' block, and at most $\cO(m)$ loops of the inner `while' block.
  Finding a shortest cycle, or all shortest $X$-paths for some $X \subseteq V(G)$, takes time $\cO(n^3)$ by the Floyd-Warshall algorithm.
  It follows that the algorithm runs in polynomial time.
  %Since such a computation is performed for each inner loop, the algorithm runs in $\cO(m^2n^3)$ time.
%
\end{proof}
}

\begin{lemma}
  \label{findboundedclusters}
  Let $d \geqslant 2$ be an integer, and let $\cP$ be a non-degenerate block-hereditary class recognizable in polynomial time.
  Then there is a polynomial-time algorithm that, given a $((\cB_{2,d} \setminus \cP) \cup \obd)$-free graph $G$, outputs the set of $(\cP \cap \cB_{2,d})$-clusters of $G$.
\end{lemma}

\appendixproof{Lemma~\ref{findboundedclusters}}
{
\begin{proof}
  Let $\cP' = \cP \cap \cB_{2,d}$.
  By \cref{finddblockalgo}, $G$ is $\cP'$-clusterable.
  We argue that the algorithm \textsc{Cluster} (\cref{clualg}) meets the requirements of the lemma.

  %%Procedure for finding clusters if we have no induced subgraph obstructions (something in \obd, or something small not in \cP)
\begin{algorithm}%[htp]
  \caption{\textsc{Cluster($G$)}}\label{clualg}
\begin{algorithmic}[1]
\Statex \textbf{Input:} A %$\cP$-clusterable graph $G$. % and an integer $d > 2$.
%$((\cB_{2,d} \setminus \cP) \cup \obd)$-free graph $G$.
$(\cB_{2,d} \setminus \cP)$-free $\cP'$-clusterable graph $G$.
\Statex \textbf{Output:} The set $\cH$ of %$\cP$-clusters of $G$.
%$(\cP \cap \cB_{2,d})$-clusters of $G$.
$\cP'$-clusters of $G$.
\State Set $\cH := \emptyset$ and $G' := G$.
\While {the shortest cycle in $G'$ has length at most $d$,} %$2d-2$,}
\State Let $X$ be the vertex set of the shortest cycle in $G'$.
\Comment {$G[X]\in \cP$}% and $\abs{X} \leq d$}
\label{foundcycle}
\While {there is a non-trivial $X$-path $P$ in $G'$ such that $\abs{X \cup V(P)} \leq d$,}%\label{nontrivialpath1}
\State Set $X := X \cup V(P)$.
\Comment {$G[X]\in \cP$}%
\label{foundpath}
\EndWhile %\label{nontrivialpath2}
%\Comment {$G[X]\in \cP$ and $\abs{X \cup V(P)} \geq d+1$ for any non-trivial $X$-path $P$} %in $G'$}
\Comment{No obstruction meets $E(G[X])$}
\State Add $G[X]$ to $\cH$.\label{foundcluster}
\State Set $G':=G'-E(G[X])$. %\label{defineg}
\EndWhile
\If {$G'$ is edgeless,}
  \State \Return $\cH$.
%\ElsIf {$K_2 \in \cP$,}
\Else
  \State \Return $\cH \cup \{G'[\{x,y\}] : \text{$x$ and $y$ are adjacent in $G'$}\}$.
%\Else
%\State \Return $\cH \cup \{G'[\{x\}] : x \in V(G') \setminus \bigcup_{H \in \cH}V(H)\}$.
\EndIf
\end{algorithmic}
\end{algorithm}

It is clear that this algorithm runs in polynomial time.
We now prove correctness of the algorithm.
%Let $\cP' = \cP \cap \cB_{2,d}$.
A graph $G[X]$ at line~\ref{foundcycle} or line~\ref{foundpath} is
$2$-connected, since it has an obvious open ear decomposition, and
in $\cP$, since $G$ is $(\cB_{2,d} \setminus \cP)$-free. % and $ \cup \obd$-free, respectively.
%So a graph $G[X]$ at line~\ref{foundcluster} is $2$-connected and in $\cP$.
By line~\ref{foundcluster}, the graph $G[X]$ has the property that,
%Since,
for any non-trivial $G[X]$ path $P$, $|X \cup V(P)| > d$.
So any $2$-connected graph containing $G[X]$ as a proper subgraph consists of at least $d+1$ vertices, and hence is not in $\cP'$.
This proves that $G[X]$ is indeed a $\cP'$-cluster.
Since $G$ is $\cP'$-clusterable, any $\cP'$-cluster distinct from $G[X]$ does not share any edges with $G[X]$, so we can safely remove them from consideration, and repeat this procedure.
If $G$ contains no cycles of length at most $d$, then the only remaining biconnected components are isomorphic to $K_1$ or $K_2$.
\end{proof}
}

\begin{theorem}\label{thm:dBGDfpt}
Let $\cP$ be a non-degenerate block-hereditary class of graphs recognizable in polynomial time.
 Then \BPBVD can be solved in time $2^{\cO(k \log d)}n^{\cO(1)}$.
\end{theorem}
\begin{proof}
  We describe a branching algorithm for \BPBVD on the instance $(G,d,k)$.
  If $G$ contains an induced subgraph in
  $(\cB_{2,d} \setminus \cP) \cup \obd$, then any solution~$S$ contains at least one vertex of this induced subgraph.
  We first run the algorithm of \cref{finddblockalgo}, and if it outputs such an induced subgraph $J$,
  then we branch on each vertex $v \in V(J)$, recursively applying the algorithm on $(G-v,d,k-1)$. 
  Since $\abs{V(J)} \leq 2d-2$, there are at most $2d-2$ branches.
 If one of these branches has a solution~$S'$, then $S' \cup \{v\}$ is a solution for $G$.
 Otherwise, if every branch returns \NO, we return that $(G,d,k)$ is a \NO-instance.
 On the other hand, if there is no such induced subgraph, then $G$ is $\cld$-clusterable, by Lemma~\ref{boundedclusterlemma},
 and we can find the set of all $\cld$-clusters in polynomial time, by Lemma%~\ref{clusterable-clusters}.
~\ref{findboundedclusters}.
  We can now run the $\cO^*(4^k)$-time algorithm of \cref{killcycles1} and return the result.  
 Thus, an upper bound for the running time is given by the following recurrence:
   $$T(n,k) =
   \begin{cases}
     1 & \text{if $k=0$ or $n=0$,} \\
     4^{k}n^{\cO(1)} & \text{if $((\cB_{2,d} \setminus \cP) \cup \obd)$-free,}\\
     (2d-2)T(n-1,k-1) + n^{\cO(1)} & \text{otherwise.}\\
   \end{cases}$$
  Hence, we have an algorithm that runs in time $\cO^*(2^{\cO(k\log d)})$.
\end{proof}

%%We remark that for the class $\cP$ of chordal or split graphs with at most $d$ vertices, the {$\cP$-Block Vertex Deletion} admits 
%%a fixed parameter tractable algorithm with the same running time in Theorem~\ref{thm:dBGDfpt} and a same lower bound in Proposition~\ref{prop:dblockdeletion-lower-bound}.

\iftoggle{paper}{
  \section{\BPBVD Lower Bounds}
  \subsection{A Tight Lower Bound}\label{sec:lower-bound}
}{
  \subsection{A Tight Lower Bound for \dBGD}\label{sec:lower-bound}
}

\newcommand{\kkC}{\textsc{$k\nobreak\times\nobreak k$~Clique}\xspace}

The Exponential-Time Hypothesis (ETH), formulated by Impagliazzo, Paturi, and Zane [10], implies that $n$-variable \textsc{3-SAT} cannot be solved in time $2^{o(n)}$.
We now argue that the previous algorithm is essentially tight under the ETH. %; that is under the assumption that \textsc{3-Sat} cannot be solved in subexponential time.

The \kkC problem takes as input an integer $k$ and a graph on $k^2$ vertices, each vertex corresponding to a distinct point of a $k$ by $k$ grid, and asks for a clique of size $k$ hitting each column of the grid exactly once.  
%%Move into proof:
%The edges between vertices of the same column cannot be involved in such a clique.
%Therefore, we may assume that each column induces an independent set.
%Finding a $k$-clique hitting each column exactly once is now equivalent to finding a $k$-clique at all. 
Unless the ETH fails, \kkC is not solvable in time $2^{o(k \log k)}$ \cite{LokshtanovMS11}.
However, solving \BPBVD in $2^{o(k \log d)}$ time, where $\cP$ contains all biconnected split graphs, implies that \kkC can be solved in $2^{o(k \log k)}$ time.
\iftoggle{paper}{}{The proof is given in the appendix.}

\lbtheorem*

\appendixproof{\cref{prop:dblockdeletion-lower-bound}}
{
\begin{proof}%[Proof of Proposition~\ref{prop:dblockdeletion-lower-bound}]
In \cite{Drange2014}, the authors show that \textsc{Component Order Connectivity} cannot be solved in $2^{o(k \log d)}$ time unless the ETH fails.
We adapt their reduction from \kkC. % which is based on the split incidence graph.
We recall that a split graph is a graph whose vertex set can be partitioned into two sets, one inducing a clique and the other inducing an independent set.
Let $(G,k)$ be an instance of \kkC.
Since the edges between vertices in the same column cannot be involved in a solution, we may assume that each column induces an independent set.
Then $(G,k)$ is a \YES-instance if and only if $G$ has a $k$-clique.
We build an instance $(G',d,k')$ of \BPBVD where $V(G')=Q \cup I$ for $Q=V(G) \cup E_1$ and $I=E_2$, where $E_1$ and $E_2$ are two copies of $E(G)$.
For each edge $e \in E(G)$, we denote by $e^1$ (resp.~$e^2$) the corresponding vertex in $E_1$ (resp.~in $E_2$).  
The set $Q$ induces a clique while $I$ induces an independent set.
For each edge $e=uv \in E(G)$, we add three edges $ue^2$, $ve^2$ and $e^1e^2$ in $G'$, each between a vertex in $Q$ and a vertex in $I$.
This ends the construction of $G'$.
Observe that $G'$ is a split graph and the vertices in $I$ all have degree~$3$.
We set $k':=k$ and $d:=|V(G')|-k-{k \choose 2}$.
Note that $|V(G')| \geqslant |V(G)| = k^2$, so $d \geqslant k^2-k-{k \choose 2}={k \choose 2}$.
Without loss of generality we may assume that $k \geqslant 3$, and hence $d \geqslant 2$. 

Assume that $G$ admits a $k$-clique $S=\{v_1,\ldots,v_k\}$, and denote $v_iv_j$ by $e_{ij}$.
We claim that $S$ is a solution for \BPBVD on the instance $(G',d,k')$.
Indeed, for each of the ${k \choose 2}$ pairs $(i,j)$ with $i<j \in [k]$, the vertex $e_{ij}^2$ has degree~$1$ in $G' - S$, so its unique neighbor, $e_{ij}^1$, is a cut vertex. Hence, $e_{ij}^2$ is not in the block containing the clique $Q \setminus S$. %, which we call the \emph{main} block.
Therefore, the blocks of $G' - S$ have %size $2$ or
at most %(if the k-clique is part of a larger clique, the block might be smaller)
$|V(G')|-k-{k \choose 2}$
vertices.

We now assume that there is a set $S \subseteq V(G')$ of at most $k$ vertices such that all the blocks of $G'-S$ have at most $d$ vertices.
%Since $G'$ is a split graph, where all the vertices of the independent set have degree $3$, it is itself a too large block.
We call the \emph{main} block the one containing $Q \setminus S$.
The only vertices of $G' - S$ that are not in the main block are the vertices $X$ of $I$ with degree at most~$1$ in $G' - S$.
Since the main block has at most $d$ vertices,
%For $S$ to be a solution,
there are at least ${k \choose 2}$ vertices in $X$.
This implies that $X$ corresponds to the ${k \choose 2}$ edges of a $k$-clique in $G$. 

Since $|V(G)|=k^2$, we have that $|E(G)|=\cO(k^4)$.
Thus $d=\cO(k^4)$ and $\log d=\cO(\log k)$.
Therefore, solving \BPBVD in $2^{o(k \log d)}$ time would also solve \kkC in time $2^{o(k \log k)}$, contradicting the ETH.  
\end{proof}
}

\iftoggle{paper}{
  \subsection{$W[1]$-hardness Parameterized Only by~$k$}\label{sec:w1hardness}

We now prove that \BPBVD is $W[1]$-hard when parameterized only by $k$, if $\cP$ is a class such that $\classp$ contains all split graphs.  In particular, this implies that \dBGD is $W[1]$-hard when parameterized only by $k$.
The reduction is similar to that in \cref{sec:lower-bound}, but the reduction is from \textsc{Clique}, rather than \kkC.

\wonehardness*
\begin{proof}
  Consider an instance $(G,k)$ of the problem \textsc{Clique} where, given a graph $G$ and integer $k$, the question is whether $G$ has a $k$-clique.
  Observe that we can perform the same reduction given in the proof of \cref{prop:dblockdeletion-lower-bound} but from \textsc{Clique}, rather than \kkC.
  By doing so, %given an instance $(G,k)$ of $\textsc{Clique}$,
  we build an instance $(G', d, k)$ of \BPBVD where $G'$ is a split graph, and for which $S$ is a solution for %\BPBVD
  the instance $(G', d, k)$
  if and only if $S$ is a $k$-clique of $G$.
  Since the reduction is parameter preserving,
  the result then follows from the fact that \textsc{Clique} is $W[1]$-hard when parameterized by the size of the solution~\cite{DowneyF13}.
\end{proof}
 
}{}

\section{$\mathcal{O}^*(c^k)$-time Algorithms %for \dCBGD and \dKBGD
Using Iterative Compression}\label{sec:single-exponential}

We now consider the specializations of \BPBVD that we refer to as \dKBGD and \dCBGD.
These problems are ``bounded'' variants of \BGVD and \DHS, respectively,
which
%The latter problems
are known to admit $c^kn^{\cO(1)}$ fixed-parameter tractable algorithms for some constant $c$.
%In particular, the authors in~\cite{Agrawal}  solve the \BGVD\ problem using a kind of clustering procedure, and then reduce to \textsc{Weighted Feedback Vertex Set}.
%It turns out that \dKBGD is hard to follow the same procedure, as we have one more constraint on the size of each block.
%Instead, we can reduce to \SFVS, but then we have a running time $2^{\mathcal{O}(k\log k)}n^{\mathcal{O}(1)}$, which is not single-exponential.
By \cref{thm:dBGDfpt}, \dKBGD and \dCBGD can be solved in $\cO^*(2^{\cO(k \log d)})$ time.
However, the next two theorems show that these problems are in fact FPT parameterized only by $k$, and, like their ``unbounded'' variants, each has a $c^kn^{\cO(1)}$-time algorithm.
%
%We show that \dCBGD and \dKBGD can be solved in time $c^k n^{\mathcal{O}(1)}$ using the well-known technique, called \emph{iterative compression}~\cite{Reed04}.
The proofs of these results use the well-known technique of iterative compression~\cite{Reed04}%
\iftoggle{paper}{, which we now briefly recap.

  %We use the iterative compression technique of Reed et al.~\cite{Reed04} and borrow ideas from known algorithms solving \FVS in time $2^{\cO(k)}$ (see, for instance, the book of Cygan et al.~\cite[Section~4.3]{Cygan15}).
In a nutshell, the idea of iterative compression is to try and build a solution of size $k$ given a solution of size $k+1$.
It is typically used for graph problems where one wants to remove a set~$S$ of at most~$k$ vertices such that the resulting graph satisfies some property or belongs to some class.
We call such a set~$S$ a \emph{solution} or a \emph{deletion set}.
Say $S$ is a solution of size $k+1$ for a problem~$\Pi$ on a graph~$G$ that we want to compress into a solution~$R$ of size at most~$k$.
We can try out all $2^{k+1}$ possible intersections of old and new solutions $I = S \cap R$.
In each case, we remove $I$ from $G$ and look for a solution of size at most $k-|I|$ that does not intersect $S \setminus I$.
We call \textsc{Disjoint}~$\Pi$ this new problem of finding a solution of size at most~$k$ that does not intersect a given deletion set~$S$ of size up to $k+1$.
If we can solve \textsc{Disjoint}~$\Pi$ in time $\cO^*(c^k)$, then the running time of this approach to solve $\Pi$ is $\cO^*(\Sigma_{i=0}^{k+1}{k+1 \choose i}c^{k-i})=\cO^*((c+1)^k)$. 
We can start with a subgraph of $G$ induced by any set of $k+1$ vertices.
Those $k+1$ vertices constitute a trivial deletion set.
After one compression step, we obtain a solution of size $k$.
Then, a new vertex is added to the graph and immediately added to the deletion set.
We compress again, and so on.
After a linear number of compressions, we have added all the vertices of $G$, so we have a solution for $G$.
For more about iterative compression, we refer the reader to Cygan et al.~\cite{Cygan15}, or Downey and Fellows~\cite{DowneyF13}. %

\subsection{\dKBGD}
}{, and are given in the appendix.

\thmdKBGDiterative*}

\appendixproof{\cref{prop:dKBGDiterative}}
{
\thmdKBGDiterative*
\begin{proof}
\iftoggle{paper}{}{%
  We use the iterative compression technique of Reed et al.~\cite{Reed04} and borrow ideas from known algorithms solving \FVS in time $2^{\cO(k)}$ (see, for instance, the book of Cygan et al.~\cite[Section~4.3]{Cygan15}).
In a nutshell, the idea of iterative compression is to try and build a solution of size $k$ given a solution of size $k+1$.
It typically works for graph problems where one wants to remove a set~$S$ of at most $k$ vertices such that the resulting graph satisfies some property or belongs to some class.
We call the set $S$ of vertices to remove the \emph{solution} or the \emph{deletion set}.
Say $S$ is a solution of size $k+1$ for a problem $\Pi$ that we want to compress into a solution $R$ of size at most $k$.
One can try out the $2^{k+1}$ potential intersections of the old and new solutions $I = S \cap R$.
In each case, one removes $I$ from the graph and looks for a solution of size at most $k-|I|$ that does not intersect $S \setminus I$.
We call \textsc{Disjoint}~$\Pi$ this new problem of finding a deletion set of size at most $k$ that does not intersect a given subset of size up to $k+1$.
If one solves \textsc{Disjoint}~$\Pi$ in time $\cO^*(c^k)$, then the running time of this approach to solve $\Pi$ would be $\cO^*(\Sigma_{i=0}^{k+1}{k+1 \choose i}c^{k-i})=\cO^*((c+1)^k)$. 
One can start with a subgraph induced by any $k+1$ vertices.
Those $k+1$ vertices constitute a trivial deletion set.
After one compression step, one gets a solution of size $k$.
Then, a new vertex is added to the graph and immediately added to the deletion set.
We compress again, and so on.
After a linear number of compressions, we have added all the vertices to the graph, and what we obtain is a solution.
One can learn more about iterative compression in these two books~\cite{Cygan15,DowneyF13}. 
}

It is sufficient to solve \textsc{Disjoint} \dKBGD in time $\cO^*(9^k)$.
Let $(G,S,d,k)$ be an instance where $S$ is a deletion set of size $k+1$.
We present an algorithm that either finds a solution $R$ of size at most $k$ not intersecting $S$, or establishes that there is no such solution.
For convenience, $G$, $S$, and $R$ are not fixed objects; they represent, respectively, the remaining graph, the set of vertices that we cannot delete, and the solution that is being built, throughout the execution of the algorithm.
Initially, $R$ is empty.
For an instance $I=(G,S,d,k)$, we take as a measure $\mu(I)=k+\cc(S)$, where %$k$ is the maximum number of vertices that can still be removed from $G$ and
$\cc(S)$ is the number of connected components of $G[S]$.
Thus, $\mu(I) \leqslant 2k+1$.
%There is no priority on which of these two rules we want to apply first.
We say a graph $G$ is a \emph{$d$-complete block graph} if every block of $G$ is a clique of size at most $d$.
We present two reduction rules and three branching rules that we apply while possible.
\begin{RULE}\label{co:deg0-or-1}
If there is a vertex $u \in V(G) \setminus S$ with degree at most $1$ in $G$, then we remove $u$ from $G$. 
\end{RULE}
The soundness of this rule is straightforward.
\begin{RULE}\label{co:direct-obstruction}
If there is a vertex $u \in V(G) \setminus S$ such that $G[S \cup \{u\}]$ is not a $d$-complete block graph, then remove $v$ from $G$, put $v$ in $R$, and decrease $k$ by $1$.
\end{RULE}
This reduction rule is safe since any induced subgraph of a $d$-complete block graph is itself a $d$-complete block graph.
Here, an \emph{obstruction} is a $2$-connected induced subgraph that is not a clique of size at most $d$.
At least one vertex of any obstruction should be in a solution.
We can restate the rule as follows:
if a vertex $u \in V(G) \setminus S$ forms an obstruction with vertices of $S$, then $u$ is in any solution.
We also observe that if a graph contains no obstruction, then it is a $d$-complete block graph.

\begin{BRANCHING}\label{co:two-obstruction}
If there are distinct vertices $u$ and $v$ in $G-S$ such that $G[S \cup \{u,v\}]$ is not a $d$-complete block graph, then branch on either removing $u$ from $G$, putting $u$ in $R$, and decreasing $k$ by $1$; or removing $v$ from $G$, putting $v$ in $R$, and decreasing $k$ by $1$.
\end{BRANCHING}
This branching rule is exhaustive since at least one of $u$ and $v$ has to be in $R$, %as mentioned in the previous paragraph.
as $G[S \cup \{u,v\}]$ contains an obstruction.
In both subinstances $\mu(I)$ is decreased by $1$,
so the %running time associated with this branching rule is $\cO^*(2^{\mu(I)})=\cO^*(4^k)$.
associated branching vector for this rule is $(1,1)$.

\begin{BRANCHING}\label{co:remove-or-connect}
If there is a vertex $u \in V(G) \setminus S$ having two neighbors $v, w \in S$ such that $v$ and $w$ are in distinct connected components of $G[S]$, then branch on %vertex $u$.
either removing $u$ from $G$, putting $u$ in $R$, and decreasing $k$ by $1$; or adding $u$ to $S$.
\end{BRANCHING}
%
%When, for some vertex $u$, we branch on either removing $u$ from $G$, putting it in $R$, and decreasing $k$ by $1$, or adding $u$ to $S$, we sometimes refer to it as \emph{branching on $u$}.
%
If $u$ is added to $S$, %(the vertices that cannot be in $R$),
then the number of connected components in $G[S]$ decreases by at least~$1$.
This branching rule is exhaustive and in both cases $\mu(I)$ is decreased by at least~$1$, %resulting in a running time of $\cO^*(2^{\mu(I)})=\cO^*(4^k)$.
so the associated branching vector is $(1,1)$.

\begin{BRANCHING}\label{co:remove-or-connect2}
Suppose there is an edge $uv$ of $G-S$ such that $u$ has a neighbor $u' \in S$ and $v$ has a neighbor $v' \in S$, and $u'$ and $v'$ are in distinct connected components of $G[S]$.
We branch on three subinstances: 
\begin{enumerate}[(a)]
  \item remove $u$ from the graph, put it in $R$, and decrease $k$ by $1$,
  \item remove $v$ from the graph, put it in $R$, and decrease $k$ by $1$, or 
  \item put both $u$ and $v$ in the set $S$. 
\end{enumerate}
\end{BRANCHING}
Again, this branching rule is exhaustive: either $u$ or $v$ is in the solution~$R$, or they can both be safely put in $S$.
In branch (c), the number of connected components of $G[S]$ decreases by at least~$1$.
Therefore, %in all three cases the measure $\mu(I)$ decreases by $1$, giving a running time of $\cO^*(3^{\mu(I)})=\cO^*(9^k)$.
the associated branching vector for this rule is $(1,1,1)$.

Applying the two reduction rules and the three branching rules presented above preserves the property that $G-S$ is a $d$-complete block graph.
The algorithm first applies these rules exhaustively (see \cref{alg:disjoint-BCBVD}), so we now assume that we can no longer apply these rules. %\cref{co:deg0-or-1,co:direct-obstruction} nor \cref{co:two-obstruction,co:remove-or-connect,co:remove-or-connect2}. 

Let $x$ be a vertex $V(G) \setminus S$ and consider its neighborhood in $S$. %, that is $N_{G[S \cup \{v\}]}(v)$.
We claim that this neighborhood is either empty, a single vertex, or all the vertices of some block of $G[S]$.
Suppose $x$ has at least two neighbors $y$ and $z$ in $S$.  Then, since \cref{co:remove-or-connect} cannot be applied, $y$ and $z$ are in the same connected component of $G[S]$.
Now, if no block of $G[S]$ contains both $y$ and $z$, then $x$ forms an obstruction with vertices in $S$, contradicting the fact that \cref{co:direct-obstruction} cannot be applied.
It follows that the vertices of the block containing $y$ and $z$, together with $x$, form a clique.  Moreover, $x$ has no other neighbors. This proves the claim.

Let $C$ be the vertex set of a leaf block of $G-S$.
We know that $C$ is a clique of size at most~$d$, and the block $G[C]$ of $G-S$ has at most one cut vertex.
If the block $G[C]$ of $G-S$ has a cut vertex $v$, let $C' := C \setminus \{v\}$; otherwise, let $C' := C$.
We use this notation in all the remaining reduction and branching rules.
%Consider the vertex set $C$ of a leaf block of %the block decomposition of
%$G-S$.
%We know that $C$ is a clique of size at most~$d$.
%Let $v$ be the unique cut vertex contained in $C$ if it exists.
%In all the remaining reduction and branching rules, if $v$ does not exist (that is, $G[C]$ is the only block of its connected component in $G-S$), replace $\{v\}$ by $\emptyset$.
The next three rules handle the case where at most one vertex in $C$ has neighbors in $S$.
\begin{RULE}\label{co:all-deg0}
If none of the vertices in $C'$ have neighbors in $S$, then remove $C'$ from $G$.
\end{RULE}
If the block $G[C]$ of $G-S$ does not have a cut vertex, then $C'$ is a connected component of $G$, so we obtain an equivalent instance after removing $C'$ from $G$.
If the block $G[C]$ of $G-S$ has a cut vertex $v$, then either $v$ is in the solution $R$, and $C'$ is a clique of size at most $d$ that is a connected component of $G-R$; or $v$ is in the $d$-complete block graph $G-R$, and $G[C]$ is a leaf block of this graph.
In either case, no vertex in $C'$ can be in an obstruction.
Each vertex not in any obstruction can be %ignored and
removed from $G$ without changing the value of %parameter
$k$.

The soundness of the next rule follows from a similar argument.

\begin{RULE}\label{co:unique-red-nocut}
  If the block $G[C]$ of $G-S$ does not have a cut vertex, $w \in C$ has at least one neighbor in $S$, and each vertex in $C \setminus \{w\}$ has no neighbor in $S$, then put $C$ in $S$.
\end{RULE}

\begin{RULE}\label{co:unique-red}
  If the block $G[C]$ of $G-S$ has a cut vertex, $w \in C'$ has at least one neighbor in $S$, and each vertex in $C \setminus \{w\}$ has no neighbor in $S$, then put $w$ in $S$.
\end{RULE}
In order to show that this rule is sound, we now prove that if $R$ is a solution containing $w$, then $(R \setminus \{w\}) \cup \{v\}$ is also a solution.
%
%At this point, \cref{co:direct-obstruction,co:remove-or-connect} cannot be applied.
%After applying this rule, $w$ is either a leaf of $G[S]$ (if it had degree $1$ in $S$) or it is not a cut vertex of a block (if its neighborhood in $S$ was an entire block).
%If $G[C]$ does not have a cut vertex, the entire set $C'$ can be put in $S$ rather than just $w$.
%Since \cref{co:two-obstruction} does not apply, %if $v$ exists,
%then any solution~$R$ containing $w$ can be transformed into another solution $R \cup \{v\} \setminus \{w\}$ (that does not contain $w$).  
%
Since \cref{co:direct-obstruction} cannot be applied, $w$ together with its neighborhood in $S$ forms a maximal clique in $G$, and this clique consists of at most $d$ vertices.
Suppose $R \setminus \{w\}$ is not a solution.  Then $G-(R \setminus \{w\})$ contains some obstruction, and any such obstruction contains $w$.  Since $w$ is not contained in a clique of size more than~$d$ in $G$, every minimal obstruction is either a diamond or an induced cycle.  Since \cref{co:two-obstruction} cannot be applied, $w$ is not contained in an induced diamond subgraph of $G$.
But every induced cycle containing $w$ and not contained in $S \cup \{w\}$ must contain $v$, which proves the claim.

Now, we consider the case where at least two vertices in $C$ each have at least one neighbor in $S$.
Let $x$ and $y$ be two such vertices.
We claim that $x$ and $y$ have the same neighborhood in $S$.
Since $C$ is a clique, $x$ and $y$ are adjacent.
As \cref{co:remove-or-connect2} does not apply, the neighbors of $x$ and $y$ in $S$ are all in the same component of $G[S]$.  
%Suppose the neighbors of $x$ and $y$ in $S$ are not all contained in some block of $G[S]$; then $x$ and $y$ together with vertices of $S$ form an obstruction, so \cref{co:two-obstruction} can be applied; a contradiction.
%We may now assume all the neighbors of $x$ and $y$ in $S$ are contained in some block $B$ of $G[S]$. If the neighborhoods of $x$ and $y$ differ, then $G[V(B) \cup \{x,y\}]$ is an obstruction, so, again, \cref{co:two-obstruction} can be applied; a contradiction.
Now, if the neighborhoods of $x$ and $y$ differ, then $G[S \cup \{x,y\}]$ contains an obstruction, so \cref{co:two-obstruction} can be applied; a contradiction.
This proves the claim.

Thus, $C$ can be partitioned into $C_1 \cup C_2$ where the vertices of $C_1$ all share the same non-empty neighborhood in $S$,
%which is either a single vertex or all the vertices of some block of $G[S]$,
while the vertices of $C_2$ have no neighbor in $S$. 
The previous three reduction rules handled the case where $|C_1| \le 1$. 
We now handle the case where $|C_1| \geqslant 2$.

\begin{RULE}\label{co:two-red-empty-c2}
If all the vertices of $C$ have the same non-empty neighborhood $A$ in $S$ (that is, $C_2=\emptyset$), then
remove any $s:=\max\{0,\abs{C'}-d+\abs{A}\}$ vertices of $C'$ from $G$, put them in $R$, decrease $k$ by $s$,
and put the remaining vertices of $C'$ in $S$.
\end{RULE}
Firstly, note that if $C$ contains a cut vertex $v$ and $|C \cup A| \leq d$, it is always better to have $v$ in $R$ rather than a vertex of $C'$. So, in this case, we can safely add the vertices of $C'$ to $S$, and, after doing so, $G[S]$ is still a $d$-complete block graph.
However, when $|C \cup A| > d$, we have to put some vertices of $C'$ in $R$, since $C \cup A$ is a clique.
As these vertices are twins (that is, they have the same closed neighborhood), it does not matter which $s$ vertices of $C'$ we choose.

In the final case, where $|C_1| \geqslant 2$ and $|C_2| \geqslant 1$, we use the following branching rule:
\begin{BRANCHING}\label{co:two-red}
Suppose there are at least two vertices of $C$ having the same non-empty neighborhood in $S$ (that is, $|C_1| \geqslant 2$) and at least one vertex of $S$ having no neighbor in $S$ (that is, $|C_2| \geqslant 1$).
We branch on two subinstances:
\begin{enumerate}[(a)]
  \item remove all the vertices of $C_2$ from $G$, put them in $R$, and decrease $k$ by $|C_2|$, or
  \item choose any vertex $x \in C_1 \cap C'$, then remove all the vertices of $C_1 \setminus \{x\}$, put them in $R$, and decrease $k$ by $|C_1|-1$.
\end{enumerate}
\end{BRANCHING}
We now argue that this branching rule is sound.
Let $x$ and $y$ be distinct vertices in $C_1$, and let $z$ be in $C_2$.
Then, for
%since any two vertices $x, y$ from $C_1$ plus any vertex $z$ from $C_2$ form an obstruction with $S$.  Indeed,
any common neighbor $t \in S$ of $x$ and $y$, the set $\{t,x,y,z\}$ induces a diamond in $G$, which is an obstruction. 
In order to eliminate all such obstructions, we must remove vertices from $G$ so that either $C_2$ is empty, as in (a), or $\abs{C_1} \le 1$, as in (b).
In case (b), we can safely pick $x \in C'$ as the vertex not added to $R$ since it is always preferable to add a cut vertex of $G-S$ in $C$ to $R$, rather than a vertex in $C'$.
In either subinstance, the measure is decreased by at least $1$, so the associated %running time is $\cO^*(2^{\mu(I)})=\cO^*(4^k)$.
branching vector is $(1,1)$.

\begin{algorithm}%[htp]
  \caption{\textsc{Disjoint-BoundedCompleteBlockVD($G,S,d,k$)}}
  \label{alg:disjoint-BCBVD}
\begin{algorithmic}[1]
\Statex \textbf{Input:} A graph $G$, a subset of vertices $S$, and two integers $d$ and $k$.
\Statex \textbf{Output:} A set $R$ of size at most $k$ such that $R \cap S = \emptyset $ and $G-R$ is a $d$-complete block graph.
\State Set $R := \emptyset$.
\While {$k \geqslant 0$ and $G-S$ is non-empty}
\If {\cref{co:deg0-or-1,co:direct-obstruction} or one of \cref{co:two-obstruction,co:remove-or-connect,co:remove-or-connect2} can be applied}
\State Apply this rule
\Else 
\State Apply one of \cref{co:all-deg0,co:unique-red-nocut,co:unique-red,co:two-red-empty-c2}, or \cref{co:two-red}
\EndIf
\EndWhile
\If {$k \geqslant 0$}
\State \Return $R$.
\EndIf     
\end{algorithmic}
\end{algorithm}

Once we apply one rule among \cref{co:all-deg0,co:unique-red-nocut,co:unique-red,co:two-red-empty-c2,co:two-red}, we check whether or not the first set of rules %(\cref{co:direct-obstruction,co:two-obstruction,co:remove-or-connect,co:remove-or-connect2})
can be applied again (see \cref{alg:disjoint-BCBVD}).

The algorithm ends when $G-S$ is empty, or if $k$ becomes negative, in which case there is no solution at this node of the branching tree.
Indeed, while $G-S$ has at least one vertex, there is always some rule to apply.
%At this point, $G[S]$ is still a $d$-complete block graph and we output the set of all the vertices that we have put in the solution along the way.
When $G-S$ is empty, $R$ is a solution for $(G,S,d,k)$.
Each reduction rule can only be applied a linear number of times since they all remove at least one vertex from $G$.
Thus, the overall running time is bounded above by the slowest branching,
namely the one with the branching vector $(1,1,1)$,
for which the running time is $\cO^*(3^{\mu(I)}) = \cO^*(9^k)$.
\end{proof}
}

\iftoggle{paper}{\subsection{\dCBGD}}{%
\thmdCBGDiterative*
}

\appendixproof{\cref{prop:dCBGDiterative}}
{

\thmdCBGDiterative*

\begin{proof}
As in the proof of \cref{prop:dKBGDiterative}, it suffices to solve \textsc{Disjoint} \dCBGD in time $\cO^*(25^k)$.

For a positive integer $d$, we say a \emph{$d$-cactus} is a graph where each block containing at least two vertices is either a cycle of length at most $d$ or an edge ($d$-cactus graphs are a subclass of cactus graphs).
%Again, the input graph is $G$ and the deletion set of size $k+1$ is $S$.
Similarly to \textsc{Disjoint} \dKBGD, we take as a measure $\mu(I)=k+\cc(S)$ for an instance $I=(G,S,d,k)$, we denote by $R$ the solution that we build, and an \emph{obstruction} is a $2$-connected induced subgraph that is not a cycle of size at most $d$. % nor an edge.
As all the vertices of $S$ will be in the graph $G-R$, there can only be a solution to \textsc{Disjoint} \dCBGD if $G[S]$ is a $d$-cactus.
Indeed, observe that any induced subgraph of a $d$-cactus is itself a $d$-cactus.
We now assume that $G[S]$ is a $d$-cactus.
As $S$ is a solution, $G - S$ is also a $d$-cactus.
We will preserve the property that the blocks of both $G[S]$ and $G[V(G) \setminus S]$ are either cycles of length at most $d$, or consist of a single edge.

%Again, we keep implicit that we remove from the graph $G$ the vertices of degree at most $1$.
We begin by applying the following four rules while possible.
The first two of these are \cref{co:deg0-or-1,co:remove-or-connect}. 
Recall that this latter branching rule is exhaustive, %it is not problem-dependent,
and its associated branching vector is $(1,1)$ for the measure~$\mu$. 

\begin{RULE}\label{cy:direct-obstruction}
If there is a vertex $u \in V(G) \setminus S$ such that $G[S \cup \{u\}]$ is not a $d$-cactus, then remove $u$ from the graph, put $u$ in $R$, and decrease $k$ by $1$.
\end{RULE}
Again, this rule is sound since any induced subgraph of a $d$-cactus is a $d$-cactus.

We define a \emph{\red} vertex as a vertex of $V(G) \setminus S$ that has at least one neighbor in $S$.
We say that two distinct \red vertices are \emph{consecutive} \red vertices if they are both contained in some block of $G-S$ and there is a path between them in which all the internal vertices have degree~$2$ in $G$.
Let $a$ and $b$ be vertices 
that are either \red, or of degree at least~$3$, or in $S$, with the additional constraint that $a$ and $b$ are not both in $S$.
A \emph{chain from $a$ to $b$}, or simply a \emph{chain}, is the set of internal vertices, all of degree~$2$, of a path between $a$ and $b$.
%where $a \notin S$ or $b \notin S$.

\begin{BRANCHING}\label{cy:remove-or-connect2}
Suppose $v$ and $w$ are consecutive \red vertices in a block of $G-S$, where $s \in S$ is a neighbor of $v$, the vertex $t \in S$ is a neighbor of $w$, and $s$ and $t$ are in distinct connected components of $G[S]$. 
Then, either
\begin{enumerate}[(a)]
  \item remove $v$ from $G$, put $v$ in $R$, and decrease $k$ by $1$; or
  \item remove $w$ from $G$, put $w$ in $R$, and decrease $k$ by $1$; or
  \item put $\{v,w\} \cup P$ in $S$, where $P$ is a chain from $v$ to $w$.
\end{enumerate}
\end{BRANCHING}
This branching rule is safe since for any solution~$R'$ that does not contain $v$ nor $w$ but contains a vertex $z$ of the chain $P$, the set $(R' \setminus \{z\}) \cup \{v\}$ is also a solution.
As we use this observation several times, we state it as a lemma.
\begin{lemma}\label{lem:chain}
If there is a solution, then there is one that does not contain any vertex of a chain.
\end{lemma}
\begin{proof}
In any solution, we may replace a vertex in a chain $P$ by the (at least) one vertex not in $S$ among the two vertices in the open neighborhood of $P$.
\end{proof}

The branching vector of \cref{cy:remove-or-connect2} is $(1,1,1)$ since in the first two cases $k$ decreases by $1$, and in the third $\cc(S)$ decreases by $1$.

Whenever these four rules cannot be applied, we claim that:
\begin{enumerate}[(1)]
  \item every vertex of $V(G) \setminus S$ has degree at most $2$ in $S$, and
  \item the neighborhood in $S$ of the set of \red vertices in a block is contained in some connected component of $G[S]$.
\end{enumerate}
Indeed, since \cref{co:remove-or-connect} is not applicable, a vertex $w$ of $V(G) \setminus S$ has neighbors in at most one connected component in $G[S]$.
Now, suppose $w$ has at least three neighbors $a$, $b$, and $c$ in the same connected component. Then $waP_1bw$ and $wbP_2cw$, where $P_1$ is an $ab$-path in $G[S]$ and $P_2$ is a $bc$-path in $G[S]$, are distinct cycles that intersect on at least $w$ and $b$. 
Hence, \cref{cy:direct-obstruction} applies for $w$; a contradiction.
%Also, if a vertex $w$ has two neighbors $a$ and $b$ in $S$, then there should be a unique path from $a$ to $b$ in $G[S]$.
%Otherwise, \cref{cy:direct-obstruction} applies for $w$.
Finally, we see that (2) holds because otherwise \cref{cy:remove-or-connect2} would apply.

Now, we assume that the first four rules do not apply (see \cref{alg:disjoint-BCGVD}).

\begin{BRANCHING}\label{cy:3red}
  Let $u$, $v$, and $w$ be distinct vertices of a block of $G-S$.
If $u$ and $v$ are consecutive \red vertices, and $v$ and $w$ are consecutive \red vertices, then branch on putting either $u$, or $v$, or $w$ into the solution~$R$.
In each case, remove the vertex from $G$ and decrease $k$ by $1$.
\end{BRANCHING}
By (2), 
%Since \cref{cy:remove-or-connect2} does not apply,
the neighbors in $S$ of $u$, $v$ and $w$ are in the same connected component of $G[S]$. 
Therefore, a set consisting of $u$, $v$, $w$, the two chains from $u$ and $v$ and from $v$ and $w$, and $S$ induces a graph that contains a subdivision of a diamond, hence is an obstruction.
By \cref{lem:chain}, we conclude that branching on the three \red vertices $u$, $v$, and $w$ is safe.
The branching vector is again $(1,1,1)$.

We now deal with leaf blocks in $G-S$ that consist of a single edge.
We call such a block a \emph{leaf edge}.
\begin{RULE}\label{cy:leaf-edge1}
Suppose $uv$ is a leaf edge in $G-S$ where
$u$ is a \red vertex of degree~$1$ in $S$, and
$v$ is either a cut vertex of $G-S$ or a \red vertex of degree~$1$ in $S$.
Then, put $u$ in $S$.
\end{RULE}
If $v$ is a cut vertex, then this rule is safe by \cref{lem:chain} since $\{u\}$ is a chain.
Otherwise, for any solution~$R'$ containing $u$, the set $(R' \setminus \{u\}) \cup \{v\}$ is also a solution.

\begin{BRANCHING}\label{cy:leaf-edge2}
Suppose $uv$ is a leaf edge in $G-S$, where $u$ and $v$ are \red vertices, and $u$ has two neighbors in $S$.
Then, branch on putting either $u$ or $v$ into the solution~$R$.
In either case, remove the vertex from $G$ and decrease $k$ by $1$.
\end{BRANCHING}
%As \cref{cy:remove-or-connect2} does not apply, $G[S \cup\{u,v\}]$ contains an obstruction.
By (2), $G[S \cup\{u,v\}]$ contains an obstruction, so this rule is safe.
The branching vector is $(1,1)$.

When these rules have been applied exhaustively, we claim that 
\begin{enumerate}
  \item[(3)] each block of $G-S$ contains at most two \red vertices, and
  \item[(4)] each leaf edge of $G-S$ has one \red vertex, and this vertex is not a cut vertex of $G-S$, and has degree~$2$ in $S$.
\end{enumerate}
The first claim follows immediately from the fact that \cref{cy:3red} cannot be applied.
Now suppose $G-S$ has a leaf edge consisting of vertices $u$ and $v$.
Since \cref{co:deg0-or-1} cannot be applied, $u$ and $v$ are either both \red, or one is \red and the other is a cut vertex of $G-S$.  If they are both \red, then either
\cref{cy:leaf-edge2} can be applied, if $u$ or $v$ has degree at least~$2$ in $S$,
or \cref{cy:leaf-edge1} can be applied, if $u$ and $v$ both have degree~$1$ in $S$; a contradiction.
So we may assume that only $u$ is \red, and $v$ is a cut vertex of $G-S$.
In this case, by (1) and since \cref{cy:leaf-edge1} cannot be applied, $u$ has degree~$2$ in $S$.

Now, we apply one of the following three rules, if possible.
%In that case, we apply the following rule if possible.
\begin{BRANCHING}\label{cy:2red}
Suppose $u$ and $v$ are distinct (consecutive) \red vertices in a leaf block that is not an edge.
If this leaf block has a cut vertex in $G-S$ distinct from $u$ and $v$, denote it by $w$. % this cut vertex which may potentially be $u$ or $v$. 
We branch on putting either $u$, or $v$, or $w$
(if the leaf block has a cut vertex in $G-S$ distinct from $u$ and $v$) into the solution~$R$ and,
in each case, remove the vertex from $G$ and decrease $k$ by $1$.
\end{BRANCHING}
Again, the graph induced by the union of the vertices of this leaf block and $S$ contains an obstruction since the leaf block is a cycle, and (2) holds. %since \cref{cy:remove-or-connect2} does not apply.
So, by \cref{lem:chain}, it is safe to branch on $u$, or $v$, or $w$.
The branching vector is $(1,1,1)$ (or $(1,1)$ if the leaf block does not have a cut vertex in $G-S$ or if $w \in \{u,v\}$).

%If none of the rules described so far apply, then each leaf block contains at most one \red vertex.
%Then, the following two reduction rules are safe.
\begin{RULE}\label{cy:cut-is-red}
  %Suppose $B$ is a leaf block of $G-S$, and $C'$ is the set of vertices of $B$ that are not cut vertices of $G-S$.
  Suppose $C'$ is the set of vertices of a leaf block of $G-S$ that are not cut vertices of $G-S$.
  If every vertex in $C'$ is not \red, then remove $C'$ from $G$.
  Otherwise, if the leaf block is a connected component of $G-S$ and exactly one vertex $u$ of $C'$ is \red, then remove $C' \setminus \{u\}$ from $G$. 
\end{RULE} 
\begin{RULE}\label{cy:red-deg1}
If there is a leaf block in $G-S$ where the unique \red vertex~$u$ is not a cut vertex of $G-S$ and $u$ has only one neighbor in $S$, then put $u$ in $S$.
\end{RULE}
The soundness of \cref{cy:cut-is-red} is straightforward.
\cref{cy:red-deg1} is safe because in any solution that contains $u$, we can remove $u$ and add the cut vertex of the leaf block instead.

We claim that, when none of the previous rules apply,
\begin{enumerate}
  \item[(5)] each leaf block of $G-S$ has one \red vertex, 
    %and this vertex has degree~$2$ in $S$.
    and this vertex is not a cut vertex of $G-S$, and has degree~$2$ in $S$.
\end{enumerate}
Indeed, the claim holds when the leaf block is a leaf edge, by (4).
Consider a leaf block of $G-S$ that is not a leaf edge.
It follows from (3) and \cref{cy:2red} that it contains at most one \red vertex.
If it contains no \red vertices, or the only \red vertex is a cut vertex of $G-S$, then \cref{cy:cut-is-red} applies; a contradiction.
So the block has one \red vertex, and this vertex is not a cut vertex of $G-S$, and has degree at most~$2$ in $S$, by (1).
Suppose it has degree~$1$ in $S$.
Since it is not a cut vertex, \cref{cy:red-deg1} applies; a contradiction.
So the vertex has degree~$2$ in $S$.
This proves (5).

Now, we try the following branching rules where $B:=B(G-S)$ is the block tree of $G-S$.
Suppose $B_1$ is a leaf of $B$, and, for some $h \ge 2$, there is a path $B_1c_1B_2c_2 \cdots B_{h-1}c_{h-1}B_h$ such that the vertices $c_1, B_2, c_2, \ldots, B_{h-1}, c_{h-1}$ have degree $2$ in $B$, the blocks $B_2, B_3, \ldots, B_{h-1}$ contain no \red vertex, and $B_h$ contains a \red vertex~$w$ such that there is a chain~$P$ linking $w$ to $c_{h-1}$.
We may assume, by (5), that $B_1$ has one \red vertex~$v$, which is distinct from $c_1$, and $v$ has two neighbors in some connected component $C$ of $G[S]$.

\begin{BRANCHING}\label{cy:path-cut-tree}
If $w$ has a neighbor in the connected component $C$, then branch on putting either $v$, or $w$, or $c_{h-1}$ into the solution~$R$.
In each case, remove the vertex from $G$ and decrease $k$ by $1$.
\end{BRANCHING}
  \begin{figure}
    \centering
{ 
        
     \subfloat[An example of \cref{cy:path-cut-tree}, with $h=4$.  The obstruction is shown in red (light gray), and $v$, $c_3$, and $w$ are the three vertices to branch on.]{
     \label{fig:rule1}
     \includegraphics[scale=0.52]{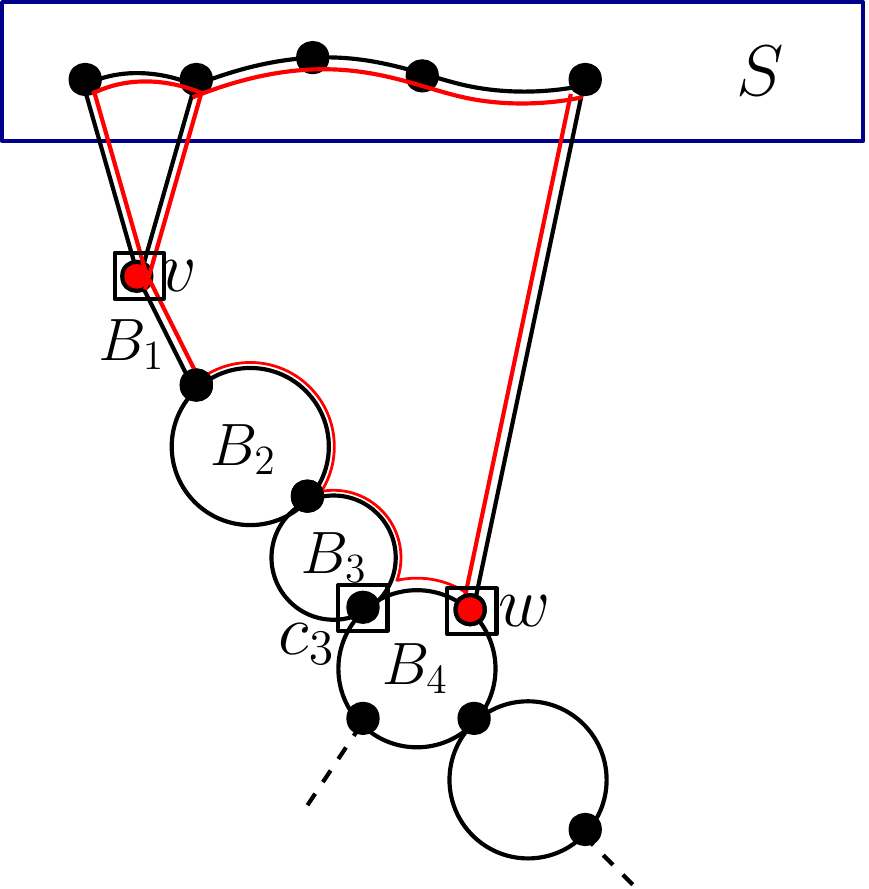}
    }
    \qquad
     \subfloat[An example of \cref{cy:branch-cut-tree2}.  Either one vertex among $v,x,y,w$ is in the solution~$R$, or %the number of connected components induced by $S$ decreases by $1$.
       all the vertices in the boxed region are not in the solution~$R$.]{
     \label{fig:rule2}       
      \includegraphics[scale=0.53]{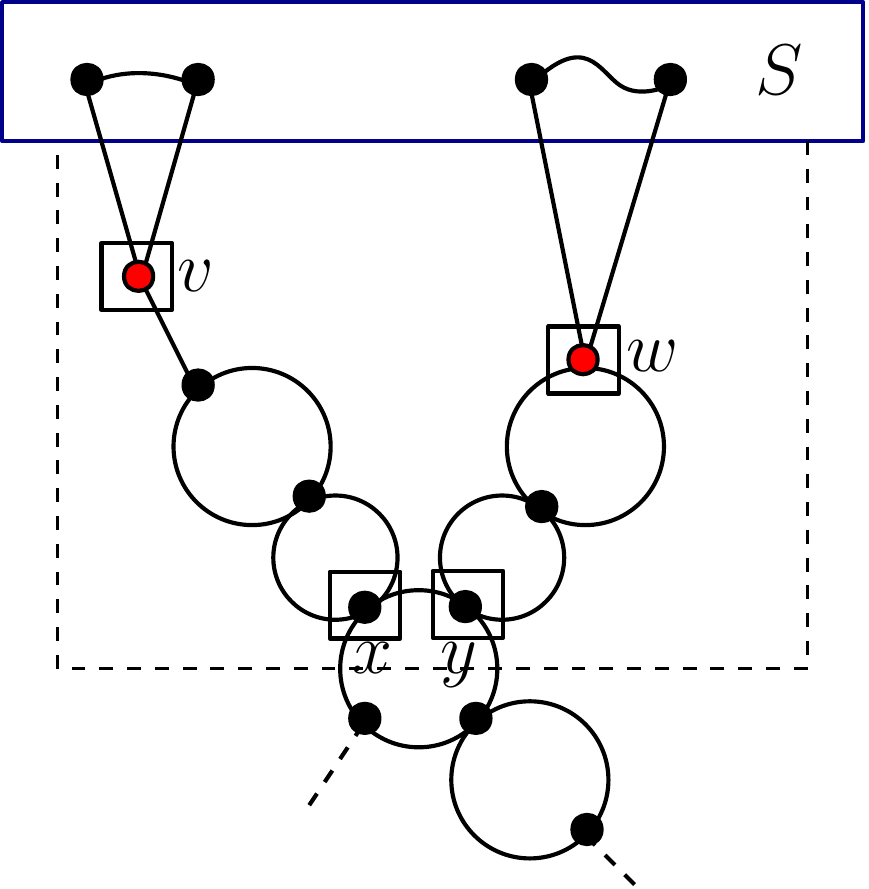}
    }
}
    \caption{Illustrations of \cref{cy:path-cut-tree,cy:branch-cut-tree2}.} \label{fig:2-rules}
  \end{figure}
To show that this rule is sound, we need a lemma analogous to \cref{lem:chain} for a ``chain of blocks'', which we now define.
Let $h \ge 3$, and let $B_1, B_2, \ldots, B_h$ be a set of blocks of $G-S$ that, except for $B_1$ and $B_h$, contain no \red vertices, and for each $i \in \{2,3,\ldots,h-1\}$, the block $B_i$ intersects two other blocks in $G-S$: $B_{i-1}$ and $B_{i+1}$.
%We further assume that $B_2$ (resp.~$B_{h-1}$) intersects a second block $B_1 \notin \{B_2, B_3 \ldots, B_{h-1}\}$ (resp.~$B_h \notin \{B_2, B_3 \ldots, B_{h-1}\}$) in $G-S$, and $B_h$ intersects exactly one other block in $G-S$ that is not in  $\{B_2, B_3 \ldots, B_{h-1}\}$.
The sequence of blocks $B_2, B_3, \ldots, B_{h-1}$ is called a \emph{chain of blocks}.% if it is maximal in satisfying all those conditions.

\begin{lemma}\label{lem:chain2}
If there is a solution, then there is one that does not contain any vertex of a chain of blocks, except potentially the cut vertex shared by $B_{h-1}$ and $B_h$.
\end{lemma}

\begin{proof}
  By hypothesis, $B_h$ shares one cut vertex $c_{h-1}$ with $B_{h-1}$.
  From any solution that contains a vertex $v$ in some block in $\{B_2, B_3, \ldots, B_{h-1}\}$, one can obtain a new solution by replacing $v$ by $c_{h-1}$.  
\end{proof}

  Recall that $P$ is a chain from $c_{h-1}$ to $w$.
The subgraph induced by $V(B_1) \cup V(B_2) \cup \cdots \cup V(B_{h-1}) \cup P \cup \{w\} \cup V(C)$ contains an obstruction.  So, by a combination of \cref{lem:chain,lem:chain2}, we can safely branch on the three vertices $v$, $w$ and $c_{h-1}$.
The branching vector is $(1,1,1)$.

\begin{BRANCHING}\label{cy:path-cut-tree2}
If $w$ has a neighbor in $S$ in a different connected component than $C$, then branch on putting either $v$, or $w$, or $c_{h-1}$ into the solution~$R$, or put $V(B_1) \cup V(B_2) \cup \cdots \cup V(B_{h-1}) \cup P \cup \{w\}$ in $S$.
In each of the first three cases, we remove the vertex from $G$ and decrease $k$ by $1$.
\end{BRANCHING}
The correctness of this rule is also based on \cref{lem:chain2}. 
In the fourth branch, $\cc(S)$ decreases by $1$, so the branching vector is $(1,1,1,1)$.

Once again, we assume that none of the previous rules apply.
Let $B'$ be a connected component of the block tree $B$ of $G-S$.
Assume $B'$ has a node that is not a leaf and let us root $B'$ at this node.
For a vertex $a$ with two neighbors $b$ and $c$, the operation of \emph{smoothing} $a$ consists of deleting $a$ and adding the edge $bc$.
%The operation of \emph{smoothing} a vertex $a$ with two neighbors $b$ and $c$ consists of deleting $a$ and adding the edge $bc$.
Let $T$ be the rooted tree that we obtain by smoothing each vertex of degree~$2$ that is not the root (and keeping the same root).
We now consider the parent node $p$ of leaves at the largest depth of $T$.
We assume that $p$ is a block $C$.
If $p$ is a cut vertex, we get a simplified version of what follows. 
As $p$ was not smoothed, it has at least two children in $T$, which,
by construction, are leaf blocks.
%Similar to consecutive \red vertices, 
We say that two distinct cut vertices $x$ and $y$ are \emph{consecutive} cut vertices if they are both contained in some block of $G-S$ and there is a chain from $x$ to $y$.
We consider two consecutive cut vertices $x$ and $y$ in $C$, % corresponding to children of $p$. % in $B$.
and let $B_x$ and $B_y$ be children of $p$ in $T$ such that $B_x\dotsm xC$ and $B_y\dotsm yC$ are paths of $B$ for which the internal vertices have degree~$2$.
Observe that we can find two consecutive cut vertices in $C$, since if one of the vertices of the $xy$-path in $C$ is \red, then either \cref{cy:path-cut-tree} or \cref{cy:path-cut-tree2} would apply.
Let $v$ (resp.~$w$) be the \red vertex in $B_x$ (resp.~$B_y$).
Recall that $v$ and $w$ have two neighbors in $S$.
We branch in the following way:

\begin{BRANCHING}\label{cy:branch-cut-tree}
If $v$ and $w$ have their neighbors in $S$ in the same connected component of $G[S]$, then branch on putting either $v$, or $w$, or $x$, or $y$ into the solution~$R$.
In each case, remove the vertex from $G$ and decrease $k$ by $1$.
\end{BRANCHING}

\begin{BRANCHING}\label{cy:branch-cut-tree2}
If $v$ and $w$ have their neighbors in $S$ in distinct connected components of $G[S]$, then branch on putting either $v$, or $w$, or $x$, or $y$ into the solution~$R$, or put all of them in $S$ together with the chain from $x$ to $y$ and all the vertices of the blocks in the path from $x$ to $B_x$ in $B$ and in the path from $y$ to $B_y$ in $B$.
In each of the first four cases, remove the vertex from $G$ and decrease $k$ by $1$.
\end{BRANCHING}
The soundness of these two branching rules is similar to \cref{cy:path-cut-tree,cy:path-cut-tree2} and relies on the fact that there is a chain from $x$ to $y$.
Their branching vectors are $(1,1,1,1)$ and $(1,1,1,1,1)$ respectively.

Finally, if none of the previous rules apply, the connected components of $G-S$ contain exactly one \red vertex.
This implies that each of these connected components is a single vertex.
Therefore, $G-S$ is an independent set and all the vertices of $V(G) \setminus S$ have exactly two neighbors in $S$.
At this point, we could finish the algorithm in polynomial time by observing that the problem is now equivalent to the problem where, given a set of paths in a forest, the task is to find a maximum-sized subset such that the paths are pairwise edge-disjoint.
An alternative is to finish with the following simple rule.

\begin{BRANCHING}\label{cy:last-step}
If $v$ and $w$ are two vertices of $V(G) \setminus S$ such that $G[S \cup \{v,w\}]$ contains an obstruction,
then branch on putting either $v$ or $w$ into the solution~$R$.
In either case, remove the vertex from $G$ and decrease $k$ by $1$.
\end{BRANCHING}
The soundness of this rule is straightforward and the branching vector is $(1,1)$.
When \cref{cy:last-step} cannot be applied, then all the remaining vertices of $V(G) \setminus S$ can safely be put in $S$.

\begin{algorithm}%[htp]
  \caption{\textsc{Disjoint-BoundedCactusGraphVD($G,S,d,k$)}}
  \label{alg:disjoint-BCGVD}
\begin{algorithmic}[1]
\Statex \textbf{Input:} A graph $G$, a subset of vertices $S$, and two integers $d$ and $k$.
\Statex \textbf{Output:} A set $R$ of size at most $k$ such that $R \cap S = \emptyset $ and $G-R$ is a $d$-cactus graph.
\State Set $R := \emptyset$.
\While {$k \geqslant 0$ and $G-S$ is non-empty}
\If {\cref{co:deg0-or-1}, or \cref{co:remove-or-connect}, or \cref{cy:direct-obstruction}, or \cref{cy:remove-or-connect2} can be applied}
\State Apply it
\ElsIf {\cref{cy:3red}, or \cref{cy:leaf-edge1}, or \cref{cy:leaf-edge2} can be applied}
             \State Apply it
             \ElsIf {\cref{cy:2red}, or \cref{cy:cut-is-red}, or \cref{cy:red-deg1} can be applied}
                    \State Apply it
                    \ElsIf {\cref{cy:path-cut-tree} or \cref{cy:path-cut-tree2} can be applied}
                           \State Apply it
                           \ElsIf {\cref{cy:branch-cut-tree} or \cref{cy:branch-cut-tree2} can be applied}
                                \State Apply it
                                %\EndIf
                           %\EndIf             
                    %\EndIf
      %\EndIf    
\EndIf
\EndWhile
\While {\cref{cy:last-step} can be applied}
       \State Apply it
\EndWhile
\If {$k \geqslant 0$}
\State \Return $R$.
\EndIf     
\end{algorithmic}
\end{algorithm}

The running time of \textsc{Disjoint} \dKBGD is given by the worst branching vector $(1,1,1,1,1)$, that is $\cO^*(5^{\mu(I)})=\cO^*(25^k)$. 
\end{proof}
}

\iftoggle{paper}{}{
When $d = |V(G)|$, these become $\cO^*(c^k)$-time algorithms for \BGVD and \DHS respectively.  In particular, the latter implies %the following:
that there
%\begin{corollary}
  %There
  is a deterministic FPT algorithm that solves \DHS, running in time $\cO^*(\cybasis^k)$. 
%\end{corollary}
}

\section{Polynomial Kernels}\label{sec:polykernel}

In this section, we prove the following:
\begin{theorem}\label{thm:mainkernel}
Let $\cP$ be a non-degenerate block-hereditary class of graphs recognizable in polynomial time.
Then \BPBVD admits a kernel with $\mathcal{O}(k^2 d^{7})$ vertices.
\end{theorem}
\iftoggle{paper}{%
Recall that,
for positive integers $x$ and $y$, we denote by $\cB_{x,y}$ the class of all biconnected graphs with at least $x$ vertices and at most $y$ vertices.
}{}
We fix a block-hereditary class of graphs $\cP$ recognizable in polynomial time.
%
%\noindent
The block tree of a graph can be computed in time $\mathcal{O}(\abs{V(G)}+\abs{E(G)})$~\cite{HopcroftT1973}.
Thus, one can test whether a given graph is in $\Phi_{\cP\cap \cB_{2,d}}$ in polynomial time.

Before describing the algorithm, we observe that there is a $(2d+6)$-approx\-i\-ma\-tion algorithm for the (unparameterized) minimization version of the \BPBVD problem.
%%Recall that $\cld$ is the class  of $2$-connected graphs with at most $d$ vertices, and $\obd$ is the class of $2$-connected graphs with between $d+1$ and $2d$ vertices.
We first run the algorithm of Proposition~\ref{finddblockalgo}. When we find an induced subgraph in $(\cB_{2,d} \setminus \cP) \cup \obd$, instead of branching on the removal of one of the vertices, we
remove all the vertices of the subgraph, then rerun the %$\cP$-clustering
algorithm.
%If the algorithm outputs the set of all $\cld$-clusters, 
%then we check whether each $\cld$-cluster is contained in $\cP$, and if not, then run the $\cld$-algorithm after removing the vertex set.
Hence, we can reduce to a $(\cP \cap \cB_{2,d})$-clusterable graph by removing at most $(2d-2)\cdot \OPT$ vertices.
Moreover, we can obtain the set of all $(\cP \cap \cB_{2,d})$-clusters using the algorithm in Lemma~\ref{findboundedclusters}. 
Arguments in the proof of Proposition~\ref{killcycles1} and the known $8$-approximation algorithm for \SFVS~\cite{EvenNZ2000} imply that there is a $(2d+6)$-approximation algorithm for \BPBVD.

\medskip 
We start with the straightforward reduction rules. Let $(G, d, k)$ be an instance of  \BPBVD.

\begin{RULE}[Component rule]\label{rule:blockcomponent} 
If $G$ has a connected component~$H\in \Phi_{\cP\cap \cB_{2,d}}$, then remove $H$.
\end{RULE}

\begin{RULE}[Cut vertex rule]\label{rule:cutvertex} 
Let $v$ be a cut vertex of $G$ such that $G-v$ contains a connected component~$H$ where $G[V(H)\cup \{v\}]$ is a block in $\cP\cap \cB_{2,d}$. 
Then remove $H$ from $G$.
\end{RULE}

%It is not hard to observe that Reduction Rules~\ref{rule:blockcomponent} and \ref{rule:cutvertex} are sound.
%

Now, we introduce a so-called bypassing rule. 
%This will be based on the approximation solution for the unparameterized version of \dBGD, which can be obtained by Theorem~\ref{thm:dBGVDapprox}. 
We first run the $(2d+6)$-approx\-i\-ma\-tion algorithm, and if it outputs a solution of more than $(2d+6)k$ vertices, then we have a \NO-instance.
Thus, we may assume that the algorithm outputs a solution of size at most $(2d+6)k$.
Let us fix such a set $U$. %Note that $\abs{U}\le (2d+6)k$.

\begin{RULE}[Bypassing rule]\label{rule:bypassing}
Let $v_1, v_2, \ldots, v_{t}$ be a sequence of cut vertices of $G-U$ with $2\le t\le d+1$, and let $B_1, \ldots, B_{t-1}$ be blocks of $G-U$ such that
\begin{enumerate}[(1)]
\item %$v_1v_2 \cdots v_t$ is an induced path in $G-U$, and,
  for each $i\in \{1, \ldots, t-1\}$, $B_i$ is the unique block of $G-U$ containing $v_i$ and $v_{i+1}$ and no other cut vertices of $G-U$;
\item $G$ has no edges between $(\bigcup_{1\le i\le t-1} V(B_i))\setminus \{v_1, v_t\}$ and $U$; and
\item $\abs{\bigcup_{1\le i\le t-1} V(B_i))}\ge d+1$.
\end{enumerate} 

\noindent
If $\bigcup_{1\le i\le t-1} V(B_i) \setminus \{v_1, \ldots, v_t\}=\emptyset$, then contract $v_1v_2$; otherwise, choose a vertex in  $\bigcup_{1\le i\le t-1} V(B_i)$ that is not a cut vertex of $G-U$, and remove it.
\end{RULE}
See \cref{bypassing} for an example application of \cref{rule:bypassing}.
Note that this rule can be applied in polynomial time using the block tree of $G-U$.
\begin{figure}
  \centering
  \subfloat{\includegraphics[scale=0.4]{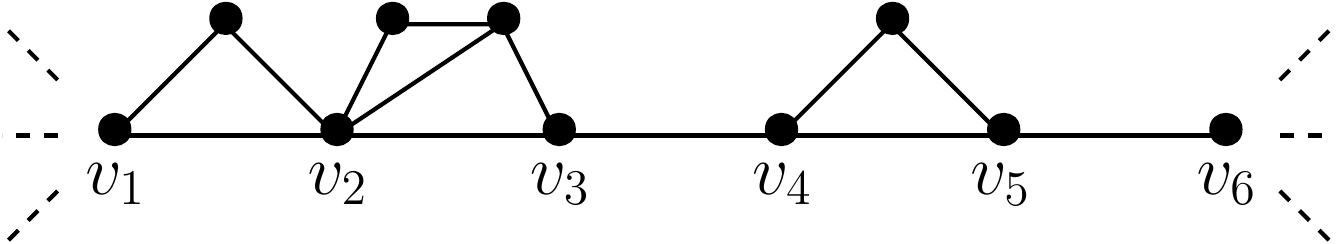}}
  \quad
  \raisebox{\baselineskip}{\Large$\rightarrow$}
  \quad
  \subfloat{\includegraphics[scale=0.4]{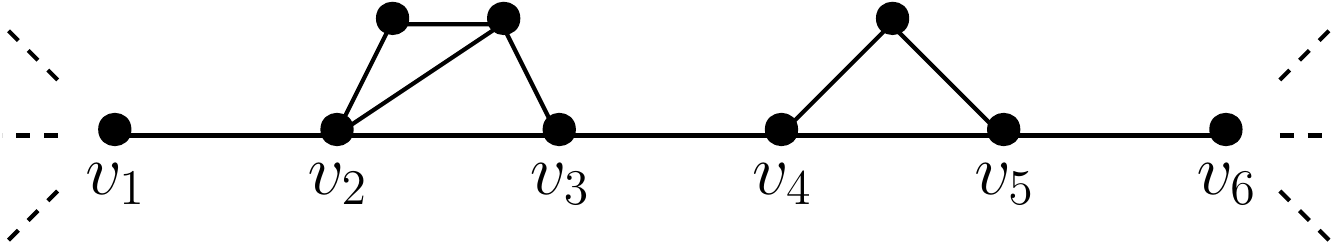}}
  \caption{An example application of \cref{rule:bypassing} when $d=9$.} \label{bypassing}
\end{figure}
\begin{lemma}\label{lem:bypassingrule}
\Cref{rule:bypassing} is safe.
\end{lemma}

%%%%
\appendixproof{Lemma~\ref{lem:bypassingrule}}
{
\begin{proof}
Let $v_1v_2 \cdots v_t$ be an induced path of $G-U$ and let $B_1, \ldots, B_{t-1}$ be blocks of $G-U$ satisfying the conditions of \cref{rule:bypassing}.
Let $G'$ be the resulting graph after applying \cref{rule:bypassing}.
We show that $G$ has a set of vertices $S$ of size at most $k$ such that $G-S\in \Phi_{\cP\cap \cB_{2, d}}$ 
if and only if $G'$ has a set of vertices $S'$ of size at most $k$ such that $G'-S'\in \Phi_{\cP\cap \cB_{2, d}}$.
For convenience, let $W:=\bigcup_{1\le i\le t-1} V(B_i)$.

Suppose that $G$ has a set of vertices $S$ of size at most $k$ such that $G-S\in \Phi_{\cP\cap \cB_{2, d}}$.
If $\abs{S\cap W}\ge 1$,
then 
$G-((S\setminus W) \cup \{v_1\})$ is also a graph in $\Phi_{\cP\cap \cB_{2, d}}$, as each block in $\{B_1, B_2, \ldots, B_{t-1}\}$ is in $\cP\cap \cB_{2, d}$.
Thus, we may assume that $S\cap  W=\emptyset$.
%There are two cases; either $v_1$ and $v_t$ are contained in the same block of $G-S$, or not.
Since $\abs{W}\ge d+1$, $v_1$ and $v_t$ are not contained in the same block of $G-S$.
It means that there is no path from $v_1$ to $v_t$ in $G-S$ containing no vertices in $\{v_2, \ldots, v_{t-1}\}$, 
and thus
all vertices in $\{v_1, v_2, \ldots, v_{t}\}$ become cut vertices of $G-S$.
Hence, all blocks in $\{B_1, B_2, \ldots, B_{t-1}\}$ are in distinct blocks of $G-S$, and thus $G'-S$ is in $\Phi_{\cP\cap \cB_{2, d}}$.

Now suppose that $G'$ has a set of vertices $S$ of size at most $k$ such that $G'-S\in \Phi_{\cP\cap \cB_{2, d}}$.
When an edge $v_1v_2$ is contracted, we label the resulting vertex $v_1$.
Similar to the other direction, if $\abs{S\cap W}\ge 1$, then we can replace $S\cap W$ with $v_1$.
So we may assume that $S\cap W=\emptyset$.
As $\abs{V(G')\cap W}\ge d$ and $G'[V(G')\cap W]$ is not $2$-connected, $v_1$ and $v_t$ cannot be contained in the same block of $G'-S$.
Thus, all blocks of $G'$ on $W$ are distinct blocks of $G'-S$, so
$G-S\in \Phi_{\cP\cap \cB_{2, d}}$.
\end{proof}
}

We show that after applying \cref{rule:blockcomponent,rule:cutvertex,rule:bypassing}, 
if the reduced graph is still large, then there is a vertex of large degree. 
This follows from the fact that the block tree of $G-U$ has no path of $2d+2$ vertices where the internal vertices have degree~$2$ in $G-U$.

\begin{lemma}\label{lem:largedegree}
Let $(G, d, k)$ be an instance reduced under \cref{rule:blockcomponent,rule:cutvertex,rule:bypassing}.
If $(G, d, k)$ is a \YES-instance and $\abs{V(G)}\ge 4d(2d+3)(d+3)k\ell$, for some integer~$\ell$, then $G$ contains a vertex of degree at least $\ell+1$.
\end{lemma}

%%%%
\appendixproof{Lemma~\ref{lem:largedegree}}
{
We first require the following lemma. % is needed to show Lemma~\ref{lem:largedegree}.

\begin{lemma}\label{lem:rootedtree}
Let $T$ be a rooted tree, and let $R\subseteq V(T)$ with $\abs{R}\ge 1$. 
Let $R'$ be the set of all nodes in $T$ that are the least common ancestor of two vertices in $R$.
Then $\abs{R'}\le  \abs{R}-1$.
\end{lemma}
\begin{proof}
The proof is by induction on $\abs{R}$.
  For each node $w$, the \emph{subtree rooted at $w$ in $T$} is the subtree of $T$ induced by $w$ and all its descendants.
Choose a minimal subtree $T'$ rooted at some $z$ containing all nodes in $R$.
If $T'$ contains precisely one node in $R$, then it contains no nodes in $R'$, by definition, and therefore $\abs{R'}=0\le \abs{R}-1$.
So we may assume that $\abs{R}\ge 2$.

Let $z_1, \ldots, z_s$ be the children of $z$ for which the subtree $T_i$ rooted at $z_i$ contains at least one node in $R$, where $i\in \{1, \ldots, s\}$.
Observe that each subtree $T_i$ satisfies $|R \cap V(T_i)| < |R|$, where, in the case that $s=1$, this is because $z \in R$ by the minimality of $T'$.
Hence, by the induction hypothesis, $\abs{R'\cap V(T_i)}\le \abs{R\cap V(T_i)}-1$ for each $i\in \{1, \ldots, s\}$.
Thus, if $s \geq 2$, \[\abs{R'}\le 1+ \bigcup_{1\le i\le s} \abs{R'\cap V(T_i)} \le 1 +  \bigcup_{1\le i\le s} ( \abs{R\cap V(T_i)}-1)\le \abs{R}-1.\]
Otherwise, $s=1$ and $z \in R$, so $|R'| \leq 1 +  ( \abs{R\cap V(T_1)}-1)\le  \abs{R}-1.$
%We conclude that $\abs{R'}\le \abs{R}-1$.
\end{proof}

\begin{proof}[Proof of Lemma~\ref{lem:largedegree}]
Suppose that $\abs{V(G)}\ge 4d(2d+3)(d+3)k\ell$ and $G$ has no vertex of degree at least $\ell+1$.
Let $T$ be the union of the block trees of connected components of $G-U$.
We color some of the nodes of $T$ as follows:
for each cut vertex $v$ of $G-U$, color $v$ red if $v$ has a neighbor in $U$; and
for each block $B$ of $G-U$, color $B$ red if $B$ contains a vertex that is not a cut vertex in $G-U$ and has a neighbor in $U$.
Observe that the number of red nodes in $T$ is at most  $2(d+3)k\ell$, since $|U| \leq 2(d+3)k$.
Arbitrarily pick a root node for each block tree.
Now, for every pair of two red vertices, color the least common ancestor in $T$ red. For all nodes that have not been colored red, color them blue. Let $R$ be the set of all red nodes in $T$. Note that $\abs{R}\le 4(d+3)k\ell$ by Lemma~\ref{lem:rootedtree}.

%First, we claim that $T-R$ has no nodes of degree at least~$3$.
%Towards a contradiction, suppose that $T'$ is a component of $T-R$ containing a node $w$ of degree at least~$3$.
%Then there are at least two components of $T'-w$ consisting of descendants of $w$ in $T'$.
%If one of the components has no red nodes, then this contradicts our assumption that $(G,d,k)$ is reduced under \cref{rule:cutvertex}.
%Thus, all the components contain red nodes, so $w$ is also colored red, by construction.

First, we claim that $T$ has no blue nodes of degree at least~$3$.
%Towards a contradiction, suppose that $T'$ is a component of $T-R$ containing a node $w$ of degree at least~$3$.
Suppose $T$ has a blue node $w$ of degree at least $3$.
Then there are at least two connected components of $T-w$ consisting of descendants of $w$ in $T$.
If one of the connected components has no red nodes, then this contradicts our assumption that $(G,d,k)$ is reduced under \cref{rule:cutvertex}.
Thus, all the connected components contain red nodes, so $w$ is also colored red, by construction.

Now, we claim that the number of connected components of $T-R$ is at most $4(d+3)k\ell$.
We obtain a forest $F$ from $T$ by contracting each maximal monochromatic subgraph $X$ of $T$ into one node with the same color as the nodes of $X$.
Note also that all leaf nodes in $F$ are colored red.
For each connected component~$F'$ of $F$, let $R'$ be the red nodes in $F'$ and let $B'$ be the blue nodes in $F'$, and arbitrarily choose a root node.  
Note that there is an injective mapping from $B'$ to $R'$
that sends a node to one of its children.
It follows that the number of blue nodes is at most the number of red nodes in $F$, 
and thus the number of connected components in $T-R$ is at most $\abs{R}\le 4(d+3)k\ell$.

Note that the number of blocks in $G-U$ is at least $\frac{\abs{V(G-U)}}{d}$, and thus
the number of nodes in $T$ is at least $\frac{\abs{V(G-U)}}{d}$.
As $\abs{V(G)}\ge 4d(2d+3)(d+3)k\ell$, there is a connected component~$B$ of $T-R$ where
\[ \abs{V(B)}\ge \frac{ \frac{\abs{V(G-U)}}{d}-\abs{R}}{4(d+3)k\ell}\ge 
 \frac{\abs{V(G-U)}-4d(d+3)k\ell}{4d(d+3)k\ell}\ge 2d+2.\]
However, this blue connected component with $2d+2$ vertices can be reduced by \cref{rule:bypassing}; a contradiction.
We conclude that if $(G,d,k)$ is a \YES-instance and $\abs{V(G)}\ge 4d(2d+3)(d+3)k\ell$, then $G$ has a vertex of degree at least~$\ell+1$.
%Since $\abs{V(G)}\ge 5d(2d+3)(d+3)^2k^2\ell^2$, 
%\[\abs{V(T)}\ge \abs{V(G)}-2(d+3)k\ell\]
\end{proof}
}

% For all graphs;
%\begin{lemma}\label{lem:largedegree}
%Let $(G, d, k)$ be an instance reduced under Reduction Rules~\ref{rule:blockcomponent}, \ref{rule:cutvertex}, and \ref{rule:bypassing}.
%If $(G, d, k)$ is a \YES-instance and $\abs{V(G)}\ge 2k(d+4)(4d\ell + 1)$, then $G$ contains a vertex of degree at least $\ell$.
%\end{lemma}

%============Sunflower

Now, we discuss a ``sunflower structure'' that allows us to find a vertex that can be safely removed.
%A similar technique was used in the quadratic kernel for \textsc{Feedback Vertex Set}~\cite{Thomasse2009}, and kernels for \BGVD~\cite{Agrawal,KimK2015}.
A similar technique was used in \cite{Agrawal,KimK2015,Thomasse2009};
there, Gallai's $A$-path Theorem is used to find many obstructions whose pairwise intersections are exactly one vertex;
%In those papers, Gallai's $A$-path Theorem is used to find many obstructions whose pairwise intersections are exactly one vertex;
here, we use different objects to achieve the same thing.

Let $A \subseteq V(G)$ and let $d \geqslant 2$.
An \emph{$(A,d)$-tree} in $G$ is a tree subgraph of $G$ on at least~$d$ vertices whose leaves are contained in $A$.
Let $v$ be a vertex of $G$. If there is an $(N_G(v), d)$-tree $T$ in $G-v$, 
then $G[V(T)\cup \{v\}]$ is a 2-connected graph with at least~$d+1$ vertices.
%Thus, we need to remove at least one vertex from $V(T)\cup \{v\}$.
This implies that if there are $k+1$ pairwise vertex-disjoint $(N_G(v), d)$-trees in $G-v$, 
then we can safely remove $v$, as any solution should contain $v$.

We prove that if $G$ does not have any set of $k+1$ pairwise vertex-disjoint $(A, d)$-trees, then 
there exists $S \subseteq V(G)$ where the size of $S$ is bounded by a function of $k$ and $d$, and every connected component of $G-S$ has fewer than $d$ vertices of $A$.
Note that $G-S$ may still have some $(A,d)$-trees, as a path of length $d-1$ between two vertices in $A$ is  also an $(A,d)$-tree.
%But for our purpose, this is sufficient.
%This follows from the fact that if a connected graph $G$ contains a set of vertices $A$ of size at least~$d$, then $G$ contains an $(A,d)$-tree.
 
\begin{proposition}\label{prop:dblockdeletion}
  Let $G$ be a graph, let $k$ and $d$ be positive integers, and let $A \subseteq V(G)$.
  There is an algorithm that, in time $\mathcal{O}(d\abs{V(G)}^3)$, finds either:
\begin{enumerate}[\rm (i)]
\item $k$ pairwise vertex-disjoint $(A,d)$-trees in $G$, or
\item a vertex subset $S \subseteq V(G)$ of size at most $2(2k-1)(d^2-d+1)$ such that each connected component of $G- S$ contains 
  fewer than $d$ vertices of $A$.\label{case2}
\end{enumerate}
\end{proposition}

%%%%
\appendixproof{Proposition~\ref{prop:dblockdeletion}}
{
%\begin{proof}
%\begin{PROP}\label{prop:dblockdeletion}
%Let $k, d$ be positive integer, and let $A$ be a vertex set of a graph $G$.
%Then in time $\mathcal{O}(d\abs{V(G)}^3)$, we can find either
%\begin{enumerate}[(1)]
%\item $k$ pairwise vertex-disjoint $(A,d)$-trees, or
%\item a vertex set $S$ of size at most $2(k+1)(d^2-d+1)$ such that each component of $G- S$ contains 
 %less than $d$ vertices of $A$.
%\end{enumerate}
%\end{PROP}

We require the following lemmas.

% and $\abs{A}\ge d$
 
\begin{lemma}\label{lem:adtreesintree}
Let $k$ and $d$ be positive integers with $d\ge 3$.
Let $T$ be a tree with maximum degree~$d$, and let $A\subseteq V(T)$.
If $\abs{A}\ge k(d^2-d+1)$, then 
there is an algorithm that finds $k$ pairwise vertex-disjoint $(A, d)$-trees in $T$, in time $\mathcal{O}(k\abs{V(T)})$.
\end{lemma}
\begin{proof}
If $k=1$, then this is trivial because $d\ge 3$. We assume that $k\ge 2$.
We choose a root node of $T$ that is not a leaf. %, and give a direction on all edges of $T$ so that there is a directed path from each node to the root.
For each node $t$ in $T$, let $w(t)$ be the number of descendants of $t$ in $A$, where $t$ is considered a descendant of itself.
We can compute the value of $w(t)$ for each $t \in T$ in time $\mathcal{O}(\abs{V(T)})$.

As $\abs{A}\ge k(d^2-d+1)\ge d$, there exists a node $t$ in $T$ with $w(t)\ge d$.
Choose such a node where $w(t_1)<d$ for every child $t_1$ of $t$.
Since $T$ has maximum degree~$d$, we have $w(t)\le d(d-1)+1=d^2-d+1$.
Clearly the subtree rooted at $t$ contains an $(A,d)$-tree.
Let $T'$ be the connected component of $T- t$ containing the parent of $t$ in $T$.
%by contracting an edge incident with $t'$ if $t'$ has degree $2$ after removing $t$.
Then $T'$ has at least $k(d^2-d+1)-(d^2-d+1)\ge (k-1)(d^2-d+1)$ nodes in $A$.
Repeating the same procedure on $T'$, we can find  $k-1$ pairwise vertex-disjoint $(A,d)$-trees in $T'$.
Thus, we can return $k$ pairwise vertex-disjoint $(A, d)$-trees of $T$ in time $\mathcal{O}(k\abs{V(T)})$.
\end{proof}

\begin{lemma}\label{lem:numberofleaves}
Let $T$ be a tree with no vertices of degree~$2$.
If $A$ is the set of all leaves of $T$, 
then $\abs{A}\ge \abs{V(T)\setminus  A}-2$.
\end{lemma}
\begin{proof}
We note that
\begin{itemize}
\item $\abs{A}+\abs{V(T)\setminus  A}=\abs{V(T)}=\abs{E(T)}-1$, and
\item $\abs{A}+ 3\abs{V(T)\setminus  A}\le \sum_{t\in V(T)}d_T(t) = 2\abs{E(T)}$.
\end{itemize}
Combining the two equations, we have that 
$\abs{A}\ge \abs{V(T)\setminus A}-2$, as required.
\end{proof}

%, and 
%\item all leaves of $H_i$ are contained in $A$, 
%\begin{itemize}
%\item 
\begin{proof}[Proof of Proposition~\ref{prop:dblockdeletion}]
We recursively construct a forest $H_i$ in $G$ such that 
each connected component of $H_i$ is an $(A, d)$-tree whose maximum degree is at most $d$,
until one of the following holds:
\begin{enumerate}[(1)]
\item $H_i$ consists of $k$ connected components.
\item $\abs{V(H_i)\cap A}\ge (2k-1)(d^2-d+1)$. 
\item For the set  $S_i$ of nodes in $H_i$ having degree other than $2$, 
every connected component of $G-((V(H_i)\cap A)\cup S_i)$ has fewer than $d$ vertices of $A$.
\end{enumerate}

\noindent
In cases (1) and (2), we will return $k$ pairwise vertex-disjoint $(A,d)$-trees,
and in case (3), we will return a set $S$ satisfying (\ref{case2}).

We start with an empty graph $H_1$.
Let $S_i$ be the set of all vertices of degree other than $2$ in $H_i$. 
For the $i$th iteration, 
choose a connected component~$C$ of $G- ((V(H_i)\cap A)\cup S_i)$ containing at least $d$ vertices of $A$.
If there is no such connected component, then we finish the procedure, as (3) holds. So assume that such a connected component~$C$ exists.
If $V(C)\cap V(H_i)\neq \emptyset$, then 
we choose a shortest path $P$ from $A\cap V(C)$ to $V(H_i)\cap V(C)$, and 
let $H_{i+1}:=H_i\cup P$.
As vertices in $A\cap V(C)$ are not contained in $H_i$, 
$\abs{V(H_{i+1})\cap A}\ge \abs{V(H_i)\cap A}+1$.
Also, the maximum degree of $H_i$ will not change as $P$ will end with a node of degree~$2$ in $H_i$.

Now, assume that $V(C)\cap V(H_i)= \emptyset$.
In this case, we find an $(A,d)$-tree in $C$ that is disjoint from $H_i$.
We choose a vertex $s\in V(C)\cap A$, and 
let $Q_1$ be the graph that consists of $s$.
For each $j\ge 2$, we recursively find a shortest path $P_i$ from $V(Q_{j-1})$ to $(V(C)\cap A)\setminus V(Q_{j-1})$ 
and let $Q_{j}:=Q_{j-1}\cup P_i$.
It is not hard to see that $Q_1$ has maximum degree~$0$, 
and for all $i\in \{2, \ldots, d\}$, $Q_i$ has maximum degree~$d$.
Also, all leaves of $Q_d$ are contained in $A$ and $\abs{V(Q_d)\cap A}=d$.
Thus, $Q_d$ is an $(A,d)$-tree.
We can compute $Q_d$ in time $\mathcal{O}(d \abs{V(G)}^2)$.
We set $H_{i+1}:=H_i\cup Q_d$.

%We choose a subset $\{s_1, \ldots, s_d\}$ of $V(C)\cap A$ with exactly $d$ vertices. 
%Starting from $s_1$ to $s_d$, 
%we construct an $(A, d)$-tree containing all vertices of $S$ as follows.

As each iteration strictly increases $V(H_i)\cap A$, 
this algorithm will terminate in at most $\abs{V(G)}$ iterations.
Let $H$ and $S$ be the final instances $H_i$ and $S_i$, respectively, prior to termination.

In case (1), each connected component of $H$ contains an $(A,d)$-tree, so
we can return $k$ pairwise vertex-disjoint $(A,d)$-trees.

Suppose we have case (2), so $\abs{V(H)\cap A}\ge (2k-1)(d^2-d+1)$.
Let $T_1, \ldots, T_h$ be the connected components of $H$.
We may assume that $h \le k-1$.
%let $P_j:=V(H_j)\cap A$ for each $1\le j\le h$.
%By the assumption, $h\le k-1$.
%By Lemma~\ref{lem:numberofleaves}, we know that the number of leaves in $G_j$ is at least $ \frac{\abs{P_j}}{2} -1$.
Applying the algorithm of Lemma~\ref{lem:adtreesintree} to $T_j$, for each $j \in \{1,\dotsc,h\}$,
we can return $\lfloor \frac{\abs{V(T_j)\cap A}}{d^2-d+1}\rfloor$ pairwise vertex-disjoint $(A, d)$-trees in time $\mathcal{O}(\frac{\abs{V(T_j)\cap A}}{d^2-d+1} \abs{V(T_j)})$.
Therefore, in this case, we can output
\[\sum_{1\le j\le h}\left( \frac{\abs{V(T_j)\cap A}}{d^2-d+1}-1 \right) \ge \frac{ (2k-1)(d^2-d+1) -h(d^2-d+1) }{d^2-d+1} \ge k\]
pairwise vertex-disjoint $(A, d)$-trees in time $\mathcal{O}(k\abs{V(G)})$.

We may now assume that case (3) holds, but case (2) does not, so $\abs{V(H)\cap A}< (2k-1)(d^2-d+1)$. By Lemma~\ref{lem:numberofleaves},
$\abs{S}\le \abs{V(H)\cap A}+2$, as all leaves of $H$ are contained in $A$.
Thus, we have
\[\abs{(V(H)\cap A)\cup S}\le 2\abs{V(H)\cap A}+2\le 2(2k-1)(d^2-d+1),\] and 
the set $(V(H)\cap A)\cup S$ satisfies (\ref{case2}).

The total running time of the algorithm is $\mathcal{O}(d\abs{V(G)}^3)$.
\end{proof}

%\end{proof}
}

\begin{RULE}[Sunflower rule~1]\label{rule:sunflower1}
Let $v$ be a vertex of $G$. 
If there are $k+1$ pairwise vertex-disjoint $(N_G(v), d)$-trees in $G-v$, then remove $v$ and reduce $k$ by $1$.
\end{RULE}

After exhaustively applying \cref{rule:sunflower1}, 
we may assume, by \cref{prop:dblockdeletion}, that for each $v\in V(G)$, there exists $S_v\subseteq V(G-v)$ with $\abs{S_v}\le 2(2k+1)(d^2-d+1)$ such that $v$ has at most $d-1$ neighbors in each connected component of $G-(S_v\cup \{v\})$.
%But it does not mean that for each component $C$ of $G-(S_v\cup \{v\})$, $G[V(C)\cup \{v\}]\in \cB_{\cld}$. 
In the remainder of this section, we use $S_v$ to denote such a set for any $v \in V(G)$.
To find many connected components of $G-(S_v\cup \{v\})$ where each connected component~$C$ has the property that $G[V(C)\cup \{v\}]\in \phi_{\cP\cap \cB_{2, d}}$, we apply the next two reduction rules.
%%For $v\in V(G)$, we define that 
%%   $$S_v =
%%    \begin{cases}
%%      S_v\cup (U\setminus \{v\}) & \text{if $v\in U$,} \\
%%      S_v\cup U & \text{otherwise.}\\
%%    \end{cases}$$
%%    Note that $\abs{S_v\cup \{v\}}\le 2(k+1)(d^2-d+1)+2(d+3)k\le 2(k+1)(d^2+4)$.
\begin{RULE}[Disjoint obstructions rule]\label{rule:disjointobs}
If there are $k+1$ connected components of $G-(S_v\cup \{v\})$ such that each connected component is not in $\phi_{\cP\cap \cB_{2, d}}$, 
then conclude that $(G,d,k)$ is a \NO-instance.
\end{RULE}

\begin{RULE}[Sunflower rule~2]\label{rule:sunflower2}
If there are $k+1$ connected components of $G-(S_v\cup \{v\})$ where each connected component~$C$ is in $\phi_{\cP\cap \cB_{2, d}}$ but $G[V(C)\cup \{v\}]\notin \phi_{\cP\cap \cB_{2, d}}$, 
then remove $v$ and decrease $k$ by $1$.
\end{RULE}

%This reduction rule is safe as each component of $G-(A_v\cup \{v\})\in \phi_{\cP\cap \cB_{2, d}}$. 
We can perform these two rules in polynomial time using the block tree of $G[V(C)\cup \{v\}]$.
Then
%By Reduction Rules~\ref{rule:disjointobs} and \ref{rule:sunflower2}, 
we may assume that $G-(S_v\cup \{v\})$ contains at most $2k$ connected components such that the connected component $C$ satisfies $G[V(C)\cup \{v\}]\notin  \phi_{\cP\cap \cB_{2, d}}$.
Thus, if $v$ has degree at least $\ell$, there are at least $\frac{\ell-2(2k+1)(d^2-d+1)}{d-1} -2k$ connected components of $G-(S_v\cup \{v\})$ such that the connected component~$C$ satisfies $G[V(C)\cup \{v\}]\in \phi_{\cP\cap \cB_{2, d}}$.
As $G$ is reduced under \cref{rule:cutvertex}, there is an edge between any such connected component~$C$ and $S_v$.
We introduce a final reduction rule, which uses the $\alpha$-expansion lemma~\cite{Thomasse2009}.

%For each vertex $v$, by Proposition~\ref{prop:dblockdeletion}, 
%there are at least $k+1$ distinct $2$-connected subgraphs with at least $d$ vertices whose pairwise intersections are exactly $v$,
%or a vertex set $S_v$ of size at most  $(k+1)(2(d-1)^2+2)$ where each component of $G- (\{v\}\cup S_v)$
%has at most $d-1$ neighbors of $v$.

%\begin{lemma}
%Let $G$ be a graph and $v\in V(G)$.
%There exists $S_v\subseteq V(G)$ with $\abs{S_v}\le$ satisfying that
%\begin{enumerate}
%\item 
%\end{enumerate} 
%\end{lemma}

%As $G-U$ is a $d$-block graph, we can take the weighted block decomposition of $G-U$.

\begin{lemma}[$\alpha$-expansion lemma]\label{lem:expansionlemma}
Let $\alpha$ be a positive integer, and let $F$ be a bipartite graph with vertex bipartition $(X, Y)$ such that $\abs{Y}\ge \alpha \abs{X}$ and every vertex of $Y$ has at least one neighbor in $X$. Then there exist non-empty subsets $X'\subseteq X$ and $Y'\subseteq Y$ and a function $\phi:X'\rightarrow \binom{Y'}{\alpha}$ such that 
\begin{itemize}
\item $N_F(Y')\cap X=X'$,
\item $\phi(x)\subseteq N_F(x)$ for each $x\in X'$, and  
\item the sets in $\{\phi(x):x\in X'\}$ are pairwise disjoint.
\end{itemize}
In addition, such a pair $X', Y'$ can be computed in time polynomial in $\alpha\abs{V(F)}$.
\end{lemma}

\begin{RULE}[Large degree rule]\label{rule:expansion}
Let $v$ be a vertex of $G$.
If there is a set $\cC$ of connected components of $G-(S_v\cup \{v\})$ such that $\abs{\cC}\ge 2d(2k+1)(d^2-d+1)$ and, for each $C\in \cC$, we have $G[V(C)\cup \{v\}]\in  \phi_{\cP\cap \cB_{2, d}}$, then do the following:
\iftoggle{paper}{%
  \begin{enumerate}[(1)]
    \item Construct an auxiliary bipartite graph $H$ with bipartition $(S_v, \cC)$ where $w\in S_v$ and $C\in \cC$ are adjacent in $H$ if and only if $w$ has a neighbor in $C$.
    \item Compute sets $\cC'\subseteq \cC$ and $S_v'\subseteq S_v$ obtained by  applying Lemma~\ref{lem:expansionlemma} to $H$ with $\alpha=d$.
    \item Remove all edges in $G$ between $v$ and each connected component~$C$ of ${\cal C'}$.
    \item Add $d-1$ internally vertex-disjoint paths of length~$2$ between $v$ and each vertex $x\in S_v'$.
    \item Remove all vertices of degree~$1$ in the resulting graph.
  \end{enumerate}
}{%
(1) Construct an auxiliary bipartite graph $H$ with bipartition $(S_v, \cC)$ where $w\in S_v$ and $C\in \cC$ are adjacent in $H$ if and only if $w$ has a neighbor in $C$. (2) Compute sets $\cC'\subseteq \cC$ and $S_v'\subseteq S_v$ obtained by  applying Lemma~\ref{lem:expansionlemma} to $H$ with $\alpha=d$.
(3) Remove all edges in $G$ between $v$ and each connected component~$C$ of ${\cal C'}$.
(4) Add $d-1$ internally vertex-disjoint paths of length~$2$ between $v$ and each vertex $x\in S_v'$.
(5) Remove all vertices of degree~$1$ in the resulting graph.
}
\end{RULE}

\begin{lemma}\label{lem:expansionsafe}
\Cref{rule:expansion} is safe.
\end{lemma}

%%%%
\appendixproof{Lemma~\ref{lem:expansionsafe}}
{
\begin{proof}
Let $\cC$ be a set of connected components of $G-(S_v\cup \{v\})$ such that $\abs{\cC}\ge 2d(2k-1)(d^2-d+1)$ and, for $C\in \cC$, $G[V(C)\cup \{v\}]\in  \phi_{\cP\cap \cB_{2, d}}$.
As $\abs{S_v}\le 2(2k-1)(d^2-d+1)$, Lemma~\ref{lem:expansionlemma} implies that
we can obtain $\cC'\subseteq \cC$, $S_v'\subseteq S_v$
and a function $\phi:S_v\rightarrow \binom{\mathcal{C}}{d}$ in polynomial time such that
\begin{itemize}
\item $N_G(\bigcup_{C\in \cC'}V(C))\cap S_v=S_v'$,  
\item $\phi(x)$ is a subset of $\mathcal{C'}$ where each connected component in $\phi(x)$ has a neighbor of $x$, and  
\item the graphs in $\{\bigcup_{C\in \phi(x)}V(C) :x\in X\}$ are pairwise disjoint.
\end{itemize}
Let $G'$ be the resulting graph obtained by applying \cref{rule:expansion}.
We prove that $(G,d,k)$ is a \YES-instance if and only if $(G', d,k)$ is a \YES-instance.
Let $R$ be the set of new vertices of degree~$2$ between $v$ and $S_v'$ in $G'$.

Suppose that $G'$ has a vertex set $A$ with $\abs{A}\le k$ such that $G'-A\in \phi_{\cP\cap \cB_{2, d}}$.
If a vertex $r\in R$ is contained in $A$ and $r'$ is a neighbor of $r$, 
then $G'-(A\setminus \{r\}\cup \{r'\})\in \phi_{\cP\cap \cB_{2, d}}$, as $r$ and all of its twins become vertices of degree~$1$ in $G'-(A\setminus \{r\}\cup \{r'\})$ and thus they cannot be contained in blocks with at least $3$ vertices.
As $d-1$ distinct paths of length~$2$ from $v$ to a vertex $x\in S_v'$ form a block with $d+1$ vertices, 
we may assume that $A$ contains either $v$ or $x$.
Considering all vertices in $S_v'$,  
we have either $v\in A$ or $S_v'\subseteq A$.
If $v\in A$, then $G-A$ is an induced subgraph of $G'-A$, and therefore, $G-A\in \phi_{\cP\cap \cB_{2, d}}$.
Suppose $v\notin A$ and $S_v'\subseteq A$. 
Since $N_G(\bigcup_{C\in \cC'}V(C))\cap S_v=S_v'$,  $v$ is a cut vertex of $G-A$, and we know that $G[V(C)\cup \{v\}]\in \phi_{\cP\cap \cB_{2, d}}$ for all $C\in \cC'$.
Moreover, $G-A-(\bigcup_{C\in \mathcal{C}} V(C))$ is an induced subgraph of $G'-A$, and thus it is a graph in $\phi_{\cP\cap \cB_{2, d}}$.
Therefore, $G-A\in \phi_{\cP\cap \cB_{2, d}}$.

For the converse direction, suppose that $G$ has a vertex set $A$ with $\abs{A}\le k$ such that $G-A\in \phi_{\cP\cap \cB_{2, d}}$.
If $v\in A$, then the vertices in $R$ become pendant vertices in $G'-A$, and thus, 
$G'-A\in \phi_{\cP\cap \cB_{2, d}}$.
We may assume that $v\notin A$.

Let $A_1:=S_v'\setminus A$ and $A_2:=A\cap  (\bigcup_{C\in \cC'}V(C))$.
It is not hard to see that $G-((A\setminus A_2)\cup A_1)\in \phi_{\cP\cap \cB_{2, d}}$ as $G[V(C)\cup \{v\}]\in \phi_{\cP\cap \cB_{2, d}}$ for each $C\in \cC'$.
We claim that $\abs{A_2}\ge \abs{A_1}$, which implies that there is a vertex set $A'$ with $\abs{A'}\le \abs{A}\le k$ such that $G'-A'\in \phi_{\cP\cap \cB_{2, d}}$. 
Suppose $\abs{A_2}<\abs{A_1}$.
Since the sets in $\{\bigcup_{C\in \phi(x)}V(C) :x\in S_v'\}$ are pairwise disjoint, 
there exists a vertex $a$ in $A_1$ such that $\phi(a)$ contains no vertex from $A_2$. 
Then $d-1$ connected components in $\phi(a)$ with the vertices $v$ and $a$ contains a $2$-connected subgraph with at least $d+1$ vertices, which contradicts the assumption that $G-A\in \phi_{\cP\cap \cB_{2, d}}$. 
%%Thus, $\abs{A_2}\ge \abs{A_1}$, and therefore $A\setminus A_2\cup A_1$ is also a proper deletion set of size at most $k$ in $G$.
%%As all vertices in $R$ become vertices of degree $1$ in $G'-(A\setminus A_2\cup A_1)$, 
%%$G'-(A\setminus A_2\cup A_1)$ is a block graph, as required.
\end{proof}
}

%%We can apply Reduction Rule~\ref{rule:expansion} in polynomial time by Lemma~\ref{lem:expansionlemma}.

\begin{lemma}
  \label{lem:polyrules}
  \Cref{rule:blockcomponent,rule:cutvertex,rule:bypassing,rule:sunflower1,rule:disjointobs,rule:sunflower2,rule:expansion} can be applied exhaustively in polynomial time.
\end{lemma}
\appendixproof{\cref{lem:polyrules}}
{
\begin{proof}
  It is clear that an application of one of \cref{rule:blockcomponent,rule:cutvertex,rule:bypassing,rule:sunflower1,rule:disjointobs,rule:sunflower2} decreases $\abs{V(G)}$. 
   We show that $\abs{V(G)}+\abs{E'}$ decreases when \cref{rule:expansion} is applied, where $E'$ is the number of edges of $G$ for which both end vertices have degree at least $3$. 

Let $(G, d, k)$ be an instance, and let $E'$ be the set of edges where both end vertices have degree at least $3$.
%We claim that $\abs{V(G)}+\abs{E'}$ decreases when Reduction Rule~\ref{rule:expansion} is applied. 
First, observe that $\abs{V(G)}+\abs{E'}$ is increased by $(d-1)\abs{S_v'}$ when adding $d-1$ disjoint paths of length~$2$ from $v$ to each vertex of $S_v'$.
It is sufficient to check that for each $C\in \cC'$, $\abs{V(G)}+\abs{E'}$ is decreased by at least $1$, %then $\abs{V(G)}+\abs{E'}$ is eventually decreased
since $\abs{\cC'}\ge d\abs{S_v'}>(d-1)\abs{S_v'}$. If $v$ has a neighbor in $C\in \cC'$ that has degree~$3$ in $G$, then this is clear. If a neighbor $w$ of $v$ in $C\in \cC'$ has degree~$2$ in $G$, then it becomes a vertex of degree~$1$ and will be removed when applying \cref{rule:expansion}. Since every neighbor of $v$ in a connected component of $\cC'$ has degree at least $2$, this completes the proof.
%As $\abs{\mathcal{C'}}\ge (d-1)\abs{S_v'}$, it is sufficient to show that for each $C\in \mathcal{C}$, $\abs{V(G)}+\abs{E'}$ is decreased by at least $1$.
%
%%Let $C\in \mathcal{C'}$. 
%%If one of the neighbors of $v$ in $C$ has degree $3$ in $G$, then 
%%$\abs{V(G)}+\abs{E'}$ is decreased by at least $1$.
%%If one of the neighbors of $v$ in $C$, say $w$, has degree $2$ in $G$, then
%%$w$ becomes a pendant vertex after removing $vw$,  and thus, we remove $w$ by Reduction Rule~\ref{rule:expansion}.
%%Therefore,   
%%$\abs{V(G)}+\abs{E'}$ is always decreased when applying Reduction Rule~\ref{rule:expansion}.
\end{proof}
}

\begin{proof}[Proof of Theorem~\ref{thm:mainkernel}]
% $\abs{V(G)}\ge 2k(d+4)(4d\ell + 1)$ 
  We apply \cref{rule:blockcomponent,rule:cutvertex,rule:bypassing,rule:sunflower1,rule:disjointobs,rule:sunflower2,rule:expansion} exhaustively.
Note that this takes polynomial time, by \cref{lem:polyrules}.
Suppose that $(G, d, k)$ is the reduced instance, and $\abs{V(G)}\ge  4dk(\ell-1)(2d+3)(d+3)$ where $\ell=2d^2(2k+1)(d^2-d+3)$.
Then, by Lemma~\ref{lem:largedegree}, there exists a vertex~$v$ of degree at least $\ell$.
%Let $S_v$ be the vertex set obtained by Proposition~\ref{prop:dblockdeletion}. 
%%$U$ be a $(2d+6)$-approximation solution obtained by the approximation algorithm for the minimization version of the \dBGD problem, and set
%% $$A_v :=
%%    \begin{cases}
%%      S_v\cup (U\setminus \{v\}) & \text{if $v\in U$,} \\
%%      S_v\cup U & \text{otherwise.}\\
%%    \end{cases}$$
    
By Proposition~\ref{prop:dblockdeletion}, $v$ has at most $d-1$ neighbors in each connected component of $G-(S_v\cup \{v\})$.
Since $\ell=2d^2(2k+1)(d^2-d+3)$, 
the subgraph $G-(S_v\cup \{v\})$ contains at least $\frac{\ell-2(2k+1)(d^2-d+1)}{d-1}\ge 2d(2k+1)(d^2-d+3)$ connected components. 
By \cref{rule:disjointobs,rule:sunflower2}, 
$G-(S_v\cup \{v\})$ contains at least $2d(2k+1)(d^2-d+1)$ connected components such that, for each connected component~$C$,
$G[V(C)\cup \{v\}]\in  \phi_{\cP\cap \cB_{2, d}}$.
Then we can apply \cref{rule:expansion}, contradicting our assumption.
We conclude that $\abs{V(G)}= \mathcal{O}(k^2d^{7})$.
\end{proof}

One might ask whether the kernel with $\mathcal{O}(k^2d^7)$ vertices can be improved upon. Regarding the $k^2$ factor, %it is probably hard to reduce to linear in $k$, as it
reducing it to linear in $k$
would imply a linear kernel for \textsc{Feedback Vertex Set}. %, which is not known. 
On the other hand, it is possible to reduce the $d^7$ factor depending on the block-hereditary class $\cP$. 
\iftoggle{paper}{}{We prove the following in the appendix.}
\begin{theorem}
  \label{improvedkernels}
  \ 
  \begin{itemize}
    \item \dBGD admits a kernel with $\cO(k^2 d^6)$ vertices.
    \item \dKBGD admits a kernel with $\cO(k^2 d^3)$ vertices.
    \item \dCBGD admits a kernel with $\cO(k^2 d^4)$ vertices. 
  \end{itemize}
\end{theorem}

\appendixproof{\cref{improvedkernels}}
{
  We prove \cref{improvedkernels} as three separate results: \cref{kerneldBGD,kerneldCBGD,kerneldKBGD}.

  First, observe that each of the three problems have $(2d+6)$-approximation algorithms, by the same argument as for the general problem  \BPBVD.
%As the algorithms in Proposition~\ref{finddblockalgo} and Lemma~\ref{findboundedclusters}
 %apply to each of the problems \dBGD, \dCBGD, and \dKBGD, 
%there are $(2d+6)$-approximation algorithms for these problems, as for the general problem  \BPBVD.

%\begin{theorem}\label{kerneldBGD}
%\dBGD admits a kernel with $\mathcal{O}(k^2d^6)$ vertices.
%\end{theorem}

Now consider the \dBGD problem.
All the reduction rules for \BPBVD can be applied with $\cP$ as the class of all biconnected graphs.
However, Reduction Rule~\ref{rule:bypassing} can be modified as follows, in order to obtain a slightly better kernel.
Let $(G, d, k)$ be an instance of \dBGD, and let $U$ be a solution of size at most $(2d+6)k$ obtained by the $(2d+6)$-approximation algorithm.

\begin{RULE}[Bypassing rule 2]\label{rule:bypassingdBGD}
Let $v_1, v_2, v_3$ be a sequence of cut vertices of $G-U$, and let $B_1$ and $B_2$ be blocks of $G-U$ such that 
\begin{enumerate}[(1)]
\item %$v_1v_2v_3$ is an induced path in $G-U$, and,
  for each $i\in \{1,2\}$, $B_i$ is the unique block containing $v_i$, $v_{i+1}$ and no other cut vertices, and
\item $G$ has no edges between $(V(B_1)\cup V(B_2))\setminus \{v_1, v_3\}$ and $U$.
\end{enumerate}
Then remove $(V(B_1)\cup V(B_2))\setminus \{v_1, v_3\}$ and add a clique of size $\min \{d, \abs{V(B_1)\cup V(B_2)}\}$ containing $v_1$ and $v_3$.
\end{RULE}
\begin{lemma}\label{lem:bypassingrule2}
\Cref{rule:bypassingdBGD} is safe.
\end{lemma}

\begin{proof}
Let $v_1v_2v_3$ be an induced path of $G-U$ and let $B_1$ and $B_2$ be blocks of $G-U$ satisfying the conditions of \cref{rule:bypassingdBGD}.
Let $G'$ be the resulting graph after applying \cref{rule:bypassingdBGD}.
We show that $G$ has a set of vertices $S$ of size at most $k$ such that $G-S\in \Phi_{\cB_{2, d}}$ 
if and only if $G'$ has a set of vertices $S'$ of size at most $k$ such that $G'-S'\in \Phi_{\cB_{2, d}}$.
For convenience, let $W:=V(B_1)\cup V(B_2)$, and
let $W'$ be the new clique added in $G'$.

Suppose that $G$ has a set of vertices $S$ of size at most $k$ such that $G-S\in \Phi_{\cB_{2, d}}$.
If $\abs{S\cap W}\ge 1$,
then 
$G-((S\setminus W) \cup \{v_1\})$ is a graph in $\Phi_{\cB_{2, d}}$, as $B_1$ and $B_2$ are in $\cB_{2, d}$.
Thus, $G'-((S\setminus W) \cup \{v_1\})$ is also a graph in $\Phi_{\cB_{2, d}}$.
We may now assume that $S\cap  W=\emptyset$.
%There are two cases; either $v_1$ and $v_t$ are contained in the same block of $G-S$, or not.
Assume that $W$ is contained in some block $B$ of $G-S$.
In this case, $\abs{W}\le d-1$ as $G[W]$ is not $2$-connected.
Thus, $G'-S\in \Phi_{\cB_{2,d}}$, as the block obtained from $B$ by replacing $W$ with $W'$ has the same number of vertices,  and all the other blocks are the same.
If $W$ is not contained in some block of $G-S$, then every path from $v_1$ to $v_3$ in $G-S$ passes through $v_2$. 
Thus, $v_1$, $v_2$, and $v_3$ are cut vertices of $G-S$.
Hence, $B_1$ and $B_2$ are distinct blocks of $G-S$, and thus $G'-S$ is in $\Phi_{\cB_{2, d}}$.

Now suppose that $G'$ has a set of vertices $S$ of size at most $k$ such that $G'-S\in \Phi_{\cB_{2, d}}$.
Similar to the other direction, if $\abs{S\cap W}\ge 1$, then we can replace $S\cap W$ with $v_1$.
So we may assume that $S\cap W=\emptyset$.
If $W'$ is a block of $G'-S$, then $v_1$ and $v_3$ are cut vertices of $G'-S$, and one can easily check that $G-S\in \Phi_{\cB_{2, d}}$.
Otherwise, the clique $W'$ is not a block of $G'-S$, that is, it is contained in a bigger block.
Then $\abs{W'}\le d-1$ and $\abs{W}=\abs{W'}\le d-1$, and thus $G-S$ is also in $\Phi_{\cB_{2,d}}$.
\end{proof}

\begin{theorem}
  \label{kerneldBGD}
  \dBGD admits a kernel with $\cO(k^2 d^6)$ vertices.
\end{theorem}
%$\abs{V(G)}\ge 4d(2d+3)(d+3)k\ell$
\begin{proof}
Lemma~\ref{lem:bypassingrule2} implies that
the block tree of $G-U$ has no path of $6$ vertices whose internal vertices have degree~$2$ in $G-U$. 
By modifying Lemma~\ref{lem:largedegree},
we can show that if $(G,d,k)$ is reduced under \cref{rule:blockcomponent,rule:cutvertex,rule:bypassingdBGD}, and $\abs{V(G)}\ge 28d(d+3)k\ell$, 
then $G$ has a vertex of degree at least $\ell+1$.
Using \cref{rule:sunflower1,rule:disjointobs,rule:sunflower2,rule:expansion} and the same argument as in the proof of \cref{thm:mainkernel}, it follows 
that there is a kernel with $\mathcal{O}(k^2d^6)$ vertices.
\end{proof}

Note that we can also use Reduction Rule~\ref{rule:bypassingdBGD} for \dKBGD, since for complete-block graphs, every maximal clique cannot be contained in a bigger block.
But it seems difficult to obtain a similar rule for \dCBGD.
%The main reason is that
%given a series of blocks that form a path in the block tree where the internal vertices have degree~$2$, it seems difficult to obtain an equivalent instance by replacing these blocks with some simpler structure.
%one cannot freely contract an edge, as the length of a cycle is a constraint, and if we replace an induced path of length $2$ with one triangle, then a new obstruction can appear.
However, for %\dCBGD and \dKBGD,
both problems,
we can obtain a smaller kernel by using different objects in Proposition~\ref{prop:dblockdeletion}.

Recall that a graph $G$ is a \emph{$d$-complete block graph} if every block of $G$ is a complete graph with at most $d$ vertices, 
and a graph $G$ is a \emph{$d$-cactus} if it is a cactus graph and every block has at most $d$ vertices.

%\begin{theorem}\label{thm:kerneldCBGD}
%\dCBGD admits a kernel with $\mathcal{O}(k^2d^3)$ vertices.
%\end{theorem}
\begin{theorem}
  \label{kerneldCBGD}
  \dCBGD admits a kernel with $\cO(k^2 d^4)$ vertices. 
\end{theorem}
\begin{proof}
We observe that for a vertex $v$ in a graph $G$ and an $(N_G(v), 3)$-tree $T$,
$G[V(T)\cup \{v\}]$ is $2$-connected and it is not a cycle. 
Thus $G[V(T)\cup \{v\}]$ is not a $d$-cactus graph, and at least one vertex of $V(T)\cup \{v\}$ should be taken in any solution.
Because of this, we can replace $(N_G(v), d)$-trees with $(N_G(v), 3)$-trees in \cref{rule:sunflower1}.
By Proposition~\ref{prop:dblockdeletion}, 
we may assume that 
for each $v\in V(G)$, there exists $S_v\subseteq V(G-v)$ with $\abs{S_v}\le 14(2k+1)$ 
such that $v$ has at most $2$ neighbors in each connected component of $G-(S_v\cup \{v\})$.

Note that if there are two vertices with three vertex-disjoint paths between them, then we have a subdivision of the diamond, which is an obstruction for cactus graphs.
Thus, the number of connected components of $G-(S_v \cup \{v\})$ required for \cref{rule:expansion} to be applicable can be changed to $(3+1) \cdot 14(2k+1)=56(2k+1)$.

We set $\ell:=112(2k+3)$. 
Suppose that $(G, d, k)$ is the reduced instance, and $\abs{V(G)}\ge  4dk(\ell-1)(2d+3)(d+3)$.
Then, by Lemma~\ref{lem:largedegree}, there exists a vertex~$v$ of degree at least $\ell$.
%Let $S_v$ be the vertex set obtained by Proposition~\ref{prop:dblockdeletion}. 
%%$U$ be a $(2d+6)$-approximation solution obtained by the approximation algorithm for the minimization version of the \dBGD problem, and set
%% $$A_v :=
%%    \begin{cases}
%%      S_v\cup (U\setminus \{v\}) & \text{if $v\in U$,} \\
%%      S_v\cup U & \text{otherwise.}\\
%%    \end{cases}$$
    
Since $\ell=112(2k+3)$, 
$G-(S_v\cup \{v\})$ contains at least $\frac{\ell-14(2k+1)}2\ge 56(2k+3)$ connected components. 
By \cref{rule:disjointobs,rule:sunflower2}, 
$G-(S_v\cup \{v\})$ contains at least $56(2k+1)$ connected components such that, for each connected component $C$,
$G[V(C)\cup \{v\}]$ is a $d$-cactus.
Then we can apply \cref{rule:expansion}, contradicting our assumption.
We conclude that $\abs{V(G)}= \mathcal{O}(k^2d^{3})$.
\end{proof}

For \dKBGD we can use Gallai's $A$-path Theorem instead of Proposition~\ref{prop:dblockdeletion}.
The following can be obtained by modifying \cite[Proposition 3.1]{KimK2015} so that the size of blocks is also taken into account.

\begin{proposition}[\cite{KimK2015}]\label{prop:generalcompletedegree}
Let $G$ be a graph and let $v\in V(G)$ and let $k$ be a positive integer.
Then, in $\mathcal{O}(kn^3)$ time, we can find either
%\begin{enumerate}[(1)]
\begin{enumerate}[\rm (i)]
  \item $k+1$ obstructions for $d$-complete block graphs that are pairwise vertex-disjoint, or\label{gcd1}
  \item $k+1$ obstructions for $d$-complete block graphs whose pairwise intersections are exactly the vertex $v$, or\label{gcd2}
  \item $S_v\subseteq V(G)$ with $\abs{S_v}\le 7k$ such that $G-S_v$ has no obstruction for $d$-complete block graphs containing $v$.\label{gcd3}
\end{enumerate}
\end{proposition}

\begin{theorem}
  \label{kerneldKBGD}
  \dKBGD admits a kernel with $\cO(k^2 d^3)$ vertices.
\end{theorem}
\begin{proof}
We exhaustively reduce
%with Proposition~\ref{prop:generalcompletedegree},
%Reduction Rule~\ref{rule:bypassingdBGD}, and other reduction rules.
  using \cref{rule:blockcomponent,rule:cutvertex,rule:bypassingdBGD,rule:sunflower1,rule:disjointobs,rule:sunflower2,rule:expansion}.
  Now, applying \cref{prop:generalcompletedegree}, we can assume that
%We can assume that 
for each $v\in V(G)$, there exists $S_v\subseteq V(G-v)$ with $\abs{S_v}\le 7k$ 
such that $G-S_v$ has no obstruction for $d$-complete block graphs containing $v$.
But %the effect of Proposition~\ref{prop:generalcompletedegree}  is different from the case of cactus graphs since
we cannot say anything about the number of neighbors of $v$ in each connected component of $G-S_v$ after reducing in case (\ref{gcd1}) or case (\ref{gcd2}).
So %, to obtain it,
we find a $(2d+6)$-approximation solution $U$ and let $U^*:=U$ if $v\notin U$ and $U^*:=U\setminus \{v\}$ otherwise, and
add it to $S_v$.
Then $G-(S_v\cup U^*)$ is a $d$-complete block graph and %thus 
$v$ has at most $d-1$ neighbors in each connected component of $G-(S_v\cup U^*)$.

Note that if there are two vertices with two vertex-disjoint paths of length at least $2$ between them, then 
there is an obstruction for $d$-complete block graphs. So we can use the $3$-expansion lemma as in \cite{KimK2015}.
Thus, the number of connected components required for \cref{rule:expansion} to be applicable can be changed to $3(7k+(2d+6)k)=3k(2d+13)$.

We set $\ell:=3kd(2d+13)$. 
Suppose that $(G, d, k)$ is the reduced instance, and $\abs{V(G)}\ge  28d(d+3)k\ell$.
By modifying Lemma~\ref{lem:largedegree},
one can show that $G$ has a vertex of degree at least $\ell+1$.

%Let $S_v$ be the vertex set obtained by Proposition~\ref{prop:dblockdeletion}. 
%%$U$ be a $(2d+6)$-approximation solution obtained by the approximation algorithm for the minimization version of the \dBGD problem, and set
%% $$A_v :=
%%    \begin{cases}
%%      S_v\cup (U\setminus \{v\}) & \text{if $v\in U$,} \\
%%      S_v\cup U & \text{otherwise.}\\
%%    \end{cases}$$
    
Since $\ell=3kd(2d+13)$, 
$G-(S_v\cup U^*\cup \{v\})$ contains at least $\frac{\ell-k(2d+13)}{(d-1)}\ge 3k(2d+13)$ connected components. 
So we can reduce the instance using the $3$-expansion lemma; a contradiction.
We conclude that $\abs{V(G)}= \mathcal{O}(k^2d^{4})$.
\end{proof}
}
%We take a $(2d+8)$-approximation solution $U$ for the non-parameterized version of the \dBGD\ problem.
%For each vertex $v$, we obtain the vertex set $S_v$ from Lemma.. 

\bibliography{main}

\iftoggle{paper}{}{
\newpage
\appendix 
\section{Appendix}
\appendixProofText
}

\end{document}